\documentclass[12pt]{article}
\newcommand{\qed}{\hfill \ensuremath{\Box}}

\newcommand{\x}{arXiv:}
\newcommand{\m}{\mathrm}

\newtheorem{proposition}{Proposition}

\newtheorem{proof}{{\em Proof}}

\usepackage[nosort]{cite}
\usepackage{graphicx}
\usepackage[font={scriptsize,singlespacing}]{caption}
\usepackage{epstopdf}
\usepackage{amsmath}
\usepackage{amssymb}
\usepackage{subfigure}
\usepackage{hyperref}
\usepackage{url}
\usepackage{xcolor}
\usepackage{color}

\def\sideremark#1{\ifvmode\leavevmode\fi\vadjust{\vbox to0pt{\vss
 \hbox to 0pt{\hskip\hsize\hskip1em
 \vbox{\hsize2cm\tiny\raggedright\pretolerance10000
 \noindent #1\hfill}\hss}\vbox to8pt{\vfil}\vss}}}%
\begin{document}
\thispagestyle{empty}
\begin{center}

\null \vskip-1truecm \vskip2truecm

{\Large{\bf \textsf{Cold Black Holes in the Harlow-Hayden Approach to Firewalls}}}

\vskip1truecm

{\large \textsf{
}}

\vskip0.1truecm
\textbf{\textsf{Yen Chin Ong}}\\
{\small \textsf{
(1) Nordita, KTH Royal Institute of Technology and Stockholm University, \\ Roslagstullsbacken 23,
SE-106 91 Stockholm, Sweden\\
(2) Leung Center for Cosmology and Particle Astrophysics \&\\ Graduate Institute of Astrophysics,  \\National Taiwan University,
Taipei 10617, Taiwan}\\
\textsf{Email: ongyenchin@member.ams.org}}\\

\vskip0.4truecm
\textbf{\textsf{Brett McInnes}}\\
{\small \textsf{Department of Mathematics,\\ National
  University of Singapore, Singapore 119076}\\
\textsf{Email: matmcinn@nus.edu.sg}}\\

\vskip0.4truecm
\textbf{\textsf{Pisin Chen}}\\
{\small \textsf{(1) Leung Center for Cosmology and Particle Astrophysics \& \\ Graduate Institute of Astrophysics \& Department of Physics,\\  National Taiwan University,
Taipei 10617, Taiwan\\
(2)  Kavli Institute for Particle Astrophysics and Cosmology, \\SLAC National Accelerator Laboratory, Stanford University, CA 94305, U.S.A} \\
\textsf{Email: pisinchen@phys.ntu.edu.tw}}\\

\end{center}
\vskip1truecm \centerline{\textsf{ABSTRACT}} \baselineskip=15pt

\medskip
Firewalls are controversial principally because they seem to imply departures from general relativistic expectations in regions of spacetime where the curvature need not be particularly large. One of the virtues of the Harlow-Hayden approach to the firewall paradox, concerning the time available for decoding of Hawking radiation emanating from charged AdS black holes, is precisely that it operates in the context of \emph{cold} black holes, which are not strongly curved outside the event horizon. Here we clarify this point. The approach is based on ideas borrowed from applications of the AdS/CFT correspondence to the quark-gluon plasma. Firewalls aside, our work presents a detailed analysis of the thermodynamics and evolution of evaporating charged AdS black holes with flat event horizons. We show that, in one way or another, these black holes are always eventually destroyed in a time which, while long by normal standards, is short relative to the decoding time of Hawking radiation.

\newpage
\addtocounter{section}{1}
\section* {\large{\textsf{1. Decoding vs. the Lifetimes of Charged Black Holes}}}

More than forty years after the discovery of Hawking radiation \cite{Hawking1, Hawking2}, its precise nature and consequences are still far from being fully understood. Although one should not discount the possibility that technological advances might lead to progress on the experimental \cite{haret, CT, Laser, Unruh, Unruh2} or observational \cite{Li} fronts in the future, at present we are forced to rely on general physical ideas in order to make progress.

Even on this front, however, there are difficulties. A basic guiding principle of quantum gravity research has long been that the quantum theory should reproduce the successes of classical General Relativity in the case of arbitrarily small spacetime curvature. The firewall controversy \cite{amps,kn:apologia} therefore threatens to develop into a serious crisis, since it seems to imply that our current general ideas regarding quantum gravitational fields will lead to theories that fail to satisfy this most basic criterion. This is because the firewall physically indicates the local presence of an event horizon, even when the associated spacetime curvature is negligible. [There are, of course, many other objections to firewalls: see for example \cite{kn:yeom}.]

There are two obvious possible fates for the information associated with a black hole: it may simply be lost \cite{kn:Wald, kn:ellis}, or it may, by some very subtle process which may or may not involve firewalls \cite{kn:stu, kn:raju, kn:HM, MS, kn:frolov} , be completely preserved.
However, it has become apparent that the ``physics of information'' \cite{kn:landauer, kn:adami, kn:reeb}, in particular the applications of information theory in gravitational physics \cite{kn:Bekenstein, kn:hp, kn:suss, kn:suss1}, may lead to other outcomes.

In particular, quantum information theory is much concerned with the \emph{time required to decode a signal}, and, in a work which has attracted much attention, Harlow and Hayden \cite{kn:HH} [see also \cite{suss2}] have proposed that this could be a key issue. The firewall argument assumes that infalling observers can make use of the information encoded in the Hawking radiation they received prior to reaching the event horizon. However, the decoding of Hawking radiation typically takes vast amounts of time, exceeding even the lifetime of an evaporating black hole\footnote{There are other effects which can be taken into account, but all of them tend to reinforce the idea that decoding of Hawking radiation may not be possible. Firstly, note that ``collecting'' Hawking radiation is not quite straightforward, as black holes radiate in all spatial directions. An infalling observer needs to devise a scheme to intercept and collate all of the Hawking radiation. It is not entirely clear that such a process is completely innocuous. It has also been argued that except for a very late and very small fraction of a black hole's lifetime, the Hawking radiation is uncorrelated with the state of the in-fallen matter \cite{BPZ}. If this is indeed the case, then an infalling observer who wishes to decode Hawking radiation will find that there is not even enough time to collect the relevant Hawking radiation [that encodes the information] before the black hole disappears. For another concern, see also \cite{IY}.}. Specifically, the conjecture in \cite{kn:HH} is that the time required is exponential in the black hole entropy. It is argued in \cite{kn:HH} that this might invalidate the firewall argument. Underlying this idea is the novel doctrine that information is truly ``physical" only if it can be \emph{decoded} [in principle].

One great advantage of the Harlow-Hayden approach is that it does not rely on understanding the precise fate of the black hole when it nears the end of evaporation. For that question is of course highly controversial: some would have it that the black hole does indeed evaporate completely, while others are willing to consider ``remnants'' \cite{kn:yak,kn:banks,kn:SH,kn:Pisin}. The emphasis in the Harlow-Hayden approach is instead on computing the time scale on which the \emph{overall} evolution occurs: one needs only to show that the \emph{longest-lived} black holes have ``short'' lifetimes when compared to the decoding time. It does not matter whether anything unusual happens at any point during this lifetime, or whether a given black hole is ``young'' or ``old''. The second great advantage is that, as we shall see, the black holes involved always have relatively low curvature outside the event horizon, so the systems we study do indeed probe precisely that regime in which the firewall argument is most controversial,\emph{ the low-curvature regime}.

Even if, as is argued in \cite{kn:apologia}, this remarkable argument does not settle the firewall problem, the idea that Hawking radiation cannot be decoded is certainly of great interest\footnote{Oppenheim and Unruh \cite{OU} recently pointed out that the Harlow-Hayden argument can be evaded by ``precomputation'' of quantum information, by forming an entangled black hole. However, this leads to superluminal signal propagation.}, and will, if correct, surely play a central role in any future complete theory of black hole evaporation. We should therefore ask:
is it really the case that black holes invariably have lifetimes shorter than the characteristic time required to decode information carried by Hawking radiation? If this is indeed so, \emph{precisely which physical effects are involved?}

Harlow and Hayden [henceforth, HH] argue that \emph{electrically charged} black holes can be expected to have lifetimes enormously longer, perhaps even infinitely longer, than their neutral counterparts: \emph{so these are the black holes that pose the most serious threat to the HH proposal}. The lifetime of a charged black hole can only be ``short'' if some additional effect intervenes.

The suggestion in \cite{kn:HH} is that string-theoretic effects, will save the day here; more precisely, that the effect known as ``\emph{AdS fragmentation}" \cite{kn:frag} destroys the charged black hole in a relatively short time. Our objective here is to be much more explicit regarding the precise nature of the physics responsible for this destruction. We shall see that the hope expressed by HH is realised in a sense [the Seiberg-Witten effect discussed below is a greatly generalized version of the effect noted in \cite{kn:frag}]; but the details are considerably more intricate than one might have expected.

It is generally accepted that the most reliable probe of quantum gravity is the AdS/CFT correspondence \cite{kn:Mal}; indeed, probably the strongest arguments in favour of the maintenance of unitarity in black hole evaporation are based on its presumed duality with a system in which unitarity is known to hold. The firewall controversy has indeed been investigated in this manner [see, for example, \cite{kn:raju,PR1,AC,MP,PR2,VV}].
However, some doubts have been raised as to whether even this powerful technique is able to deal with all aspects of black hole physics [see \cite{kn:aron} for a recent example]. It is therefore prudent to rely on some specific form of the duality which is known to work particularly well, especially when applied to \emph{charged} black holes.

There is in fact an extensive field of research in which the physics of electrically charged AdS black holes plays a central role: the application of AdS/CFT duality to the study of the Quark-Gluon Plasma [QGP]. Specifically, AdS-Reissner-Nordstr$\ddot{\m{o}}$m black holes \emph{with toral or planar event horizons} are dual to a field theory which describes a system that in many ways resembles a quark-gluon plasma inhabiting a locally flat spacetime at conformal infinity \cite{kn:solana,kn:pedraza,kn:youngman,kn:gubser,kn:janik}. [In $(n+2)$-dimensions, ``planar'' refers to event horizons with $\mathbb{R}^n$ topology; ``toral'' to the topology of the $n$-dimensional flat torus.] This version of the duality has enjoyed substantial successes [particularly with regard to the celebrated ``KSS bound'' \cite{kn:son,kn:muller}], and can claim to have some measure of experimental support. \emph{We propose to use ideas suggested by this theory to throw some light on the fate of electrically charged black holes, as they appear in the HH argument.}

When we do this, we find some unexpected answers. In particular, the duality suggests that something dramatic must happen to AdS black holes at \emph{low} temperatures [which entails low \emph{curvatures} outside the event horizon]. For, of course, the QGP cannot be expected to exist at arbitrarily low temperatures ---$\,$ it either hadronizes or undergoes a phase transition to some other, radically different [for example, ``quarkyonic''] state. This is seen in the large literature [for example \cite{kn:ohnishi,kn:mohanty,kn:satz,kn:ARSS,kn:Endorni}] devoted to the \emph{quark matter phase diagram}, which represents various states of quark matter as a function of temperature and quark chemical potential: see\footnote{In Figure 1, ALICE, RHIC, FAIR, and NICA refer to various current and projected experimental programmes \cite{kn:ilya,kn:dong,kn:fair,kn:nica,kn:rhic} designed to explore the physics of this diagram. Astrophysical phenomena such as core-collapse supernovae and neutron star mergers could also serve as arenas to study QCD phase transitions; see for example \cite{Burgio, Sagert, PisinLabun, Sasaki}.} Fig.(\ref{1}).

\begin{figure}[!h]
\centering
\includegraphics[width=1.00\textwidth]{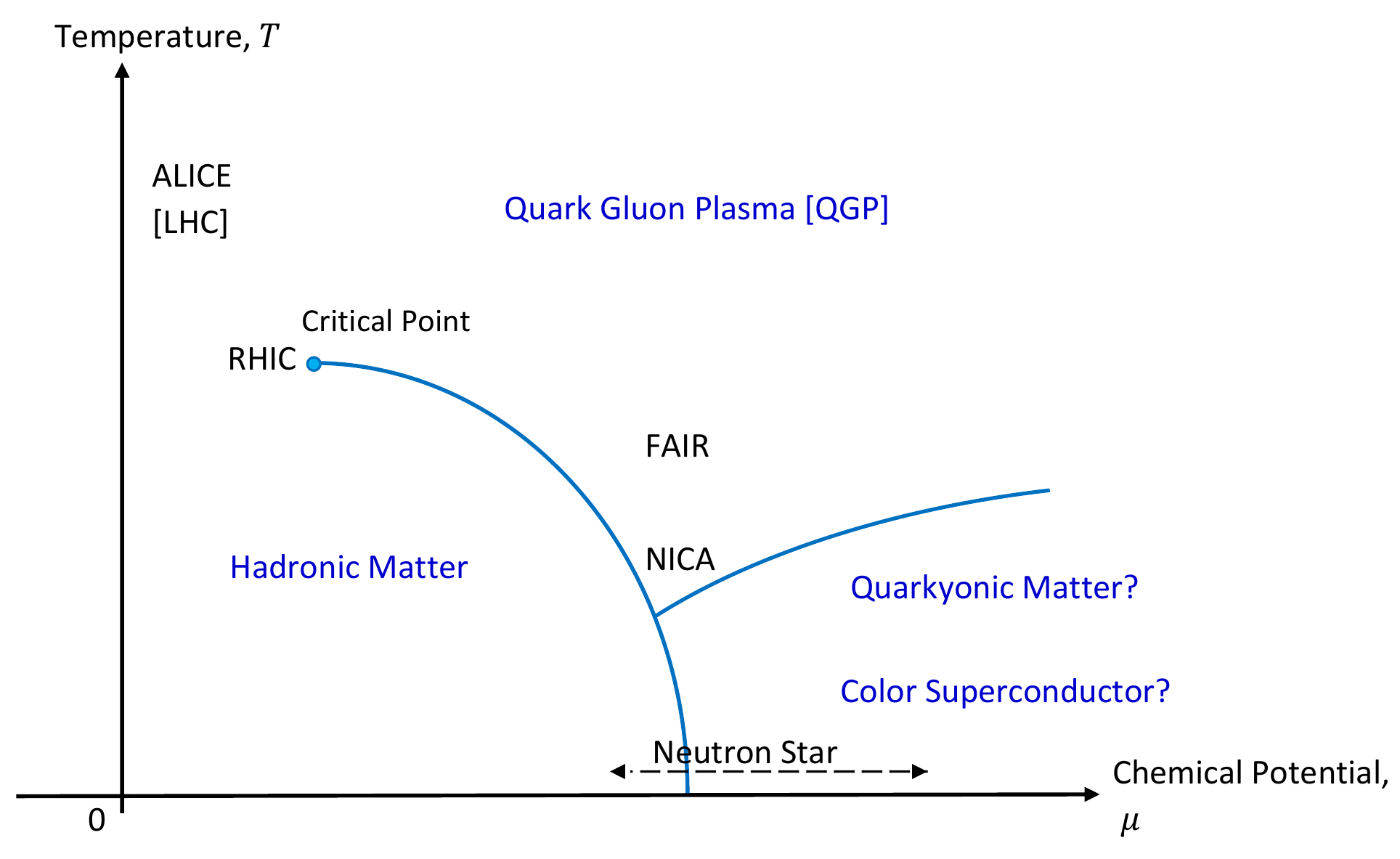}
\caption{Conjectured Quark Matter Phase Diagram. \label{1}}
\end{figure}

While many of the details remain conjectural, there is no suggestion that the plasma phase extends downward to very low temperatures, at any value of the chemical potential. In short, duality teaches us that we should expect Reissner-Nordstr$\ddot{\m{o}}$m black holes to have their lifetimes terminated by some effect which disrupts them at \emph{low} temperatures\footnote{We are of course restricting ourselves here to those AdS-Reissner-Nordstr$\ddot{\m{o}}$m black holes which do in fact give a broadly correct dual representation of the quark-gluon plasma. That immediately \emph{excludes} black holes with topologically spherical event horizons, precisely because these do suggest that the plasma phase extends down to arbitrarily low temperatures [see \cite{kn:clifford}, page 465]. Fortunately, the AdS-Reissner-Nordstr$\ddot{\m{o}}$m black holes with toral or planar event horizons, those which are in fact the ones used in applications of holography to quark matter, do not suffer from this defect \cite{kn:AdSRN}. Henceforth \emph{we confine attention to these black holes}. Note that HH do discuss [\emph{neutral}] toral black holes.}.

The reader may protest at this point: a plasma, left to its own devices, cools extremely rapidly ---$\,$ how, therefore, can this argument be relevant to black hole evaporation, which is normally taken to proceed in the opposite direction along the temperature axis? In fact, however, such behaviour for an evaporating black hole is not generic, in the following sense: the temperature of a typical [that is, with charge not exactly zero, and not already cold] black hole actually \emph{drops} initially as it evaporates. Let us explain this crucial point. [In this discussion, until further notice, we will consider the asymptotically flat case, in which the event horizon necessarily has spherical topology \cite{Hawking3, HawkingEllis}.]

When it was realised that Hawking evaporation can change the parameters of a black hole, it immediately became apparent that this posed a threat to cosmic censorship \cite{Penrose}. For clearly, if a charged or rotating black hole does not lose its charge or angular momentum at least as rapidly as it loses mass, then it is in danger of passing through one or both of the extremal limits defined by the Kerr-Newman geometry. In a classic work, Page \cite{kn:page} showed that an asymptotically flat rotating [uncharged] black hole always loses angular momentum more rapidly than it loses mass, so that censorship is safeguarded.

The charged case proved to be more difficult, and was not settled until Hiscock and Weems \cite{kn:HW} carried out a thorough numerical investigation [see also \cite{kn:SP}]. They found that, \emph{initially}, a black hole with a small but non-zero charge-to-mass ratio $Q/M$ ---$\,$ recall that the temperature is inversely related to $Q/M$, so this means that the black hole is not unusually cold at the outset ---$\,$ actually loses mass more rapidly than it loses charge as it evaporates. The temperature therefore drops, and the black hole can come quite close to extremality. However, at a certain point before that happens, the temperature reaches a non-zero minimum\footnote{There is a large literature [for recent examples, see \cite{kn:reall} and the references therein] on the question as to whether exactly extremal semi-classical black holes can exist. Note that this is not useful to us here: we need to exclude black holes with temperatures that are ``low'', \emph{not necessarily exactly zero}.}, the process reverses, and the temperature begins to rise, eventually to the arbitrarily high values made familiar by the evaporation of a Schwarzschild black hole. Censorship is again respected, but not in the simple manner of the rotational case: censorship violation is staved off ``at the last moment''. [Highly charged black holes, that is, holes which are already cold, behave more conventionally: they simply get hotter.]

In short, a generic charged asymptotically flat black hole cools at first; \emph{if it survives this cooling}, it then gets hot.
We shall see in this work that AdS-Reissner-Nordstr$\ddot{\m{o}}$m black holes with flat event horizons also undergo an initial drop in temperature; however, the numerical data strongly indicate that, in this case, the temperature \emph{always} falls, ultimately to arbitrarily small [but positive] values if no other effect disrupts the black hole\footnote{A similar pattern is observed in \cite{kn:kimwen}, even in the asymptotically flat case, though the physical argument there is very different to the one in this work.}. Charged black holes in AdS with flat event horizons, then, do behave in a manner consistent with the dual representation in terms of a cooling plasma. [Of course, a real plasma cools enormously more rapidly than the black holes considered here. \emph{Both} scales are however negligible when compared with the decoding time; showing this is our main objective.]

We can now resume the argument we were making above. In short, the evaporation of a generic charged AdS black hole with a flat event horizon causes the temperature to drop. But if the black hole becomes sufficiently cold, then it must cease to exist as a black hole, just as the dual plasma must cease to exist as a plasma as it cools. Thus the lifetime must be cut short, as HH require. We stress again that the black holes in our analysis are such that the curvature outside the event horizon is always small [around $144/L^4$, where $L$ is the asymptotic AdS curvature scale], so we are directly probing the low-curvature regime where firewalls are supposed to arise.

A less agreeable aspect of the argument, thus far, is that it is like an existence proof. It convinces us that something happens to the black hole as it cools, but it does not explain what that might be; and indeed the effect must be an unusual one, since we are more accustomed to quantum gravity effects becoming important at high, not low, temperatures and curvatures. Our objective in this work is to remedy this. That is, we wish to answer the question: \emph{exactly which physical effect is responsible for the destruction of event horizons as charged black holes cool?}

Of course, the idea that AdS black holes undergo drastic changes as they cool is very familiar. Hawking and Page \cite{kn:hawpa} showed that, in the case of spherical event horizons, there is a phase transition, and a similar statement is true in the case of toral [that is, flat but \emph{compact}] event horizons \cite{kn:surya}. In both cases the black hole ceases to exist; the cold phase \cite{kn:hormy} has a definite geometry, but it is not that of a black hole. [It is thought \cite{kn:confined} that this transition might give a simplified model of the hadronization of the QGP, at small chemical potentials.] This is what we need here, in order to complete the Harlow-Hayden argument\footnote{This is unlike the case of a holographic superconductor, where the effect of the transition \cite{kn:gubserabelian} is not to destroy the black hole but merely to cause it to grow ``hair'' [see however \cite{kn:zero}]. Notice too that this only occurs in response to the presence of a specific form of matter [usually a scalar field], whereas here we want it to occur for pure AdS-Reissner-Nordstr$\ddot{\m{o}}$m geometry.}.

However, we should not expect that the Hawking-Page transition is the \emph{only} effect responsible for the destruction of AdS black holes as they cool; this for two reasons.

First, while it is true that AdS black holes with flat \emph{compact} event horizons undergo a phase transition, this is \emph{not} true when the compactification scale is taken to infinity ---$\,$ that is, when we turn to \emph{planar} [rather than toral] black holes. [The transition temperature drops to zero in this case.] Hence charged \emph{planar} black holes do apparently have arbitrarily long lifetimes.

The second reason is revealed by Figure 1, which shows that the transition from the QGP state takes various forms, depending on the value of the chemical potential. For example, it has been suggested \cite{kn:andronic} that, at sufficiently high values of the chemical potential, the transition is not to the hadronic state but rather to a ``quarkyonic'' form of quark matter. A holographic account of this state is available \cite{kn:deboer}. There is no reason to think that the transition to this state is triggered by the same effect that causes the very different transition to the hadronic state. Therefore, we should in general expect to find that some other effect, apart from the Hawking-Page transition, is in some cases responsible for the disappearance of cold AdS black holes.

In short, then, we need to identify some novel effect which supplements the Hawking-Page transition in some cases, and which can, in particular, destabilize an AdS-Reissner-Nordstr$\ddot{\m{o}}$m black hole when its event horizon has \emph{either} toral \emph{or} planar topology, and when its temperature is low but not zero.

Just such an effect was found in \cite{kn:AdSRN}: Seiberg-Witten instability \cite{kn:seiberg} [see also \cite{KPR, Barbon}]. Seiberg and Witten showed that the stability of branes propagating in asymptotically AdS spacetimes depends on the way the ambient geometry affects the areas and volumes of the branes. For geometries with flat metrics at infinity, such as we have in the case of AdS-Reissner-Nordstr$\ddot{\m{o}}$m toral and planar black holes, the competition between the positive and negative terms in the brane action is particularly close. It turns out that the addition of small amounts of electric charge to a black hole with a flat event horizon has no ill-effects, that is, the brane action remains positive everywhere. But [for four-dimensional black holes] when the charge reaches about 92$\%$ of the extremal charge \cite{universal} ---$\,$ that is, \emph{when the temperature is low, but not zero} ---$\,$ the brane action becomes negative at a certain distance from the black hole, triggering a pair-production instability. In short, we have exactly what we need, supplied by basic objects in string theory.

In summary, we claim that AdS-Reissner-Nordstr$\ddot{\m{o}}$m toral \emph{and} planar black holes are indeed destroyed as they evaporate, as HH require; and that we can, for various values of the compactification parameter, identify the physical mechanism responsible: in the toral case it is a phase transition of the Hawking-Page type at low values of the chemical potential, the Seiberg-Witten effect at high values. [In the planar case, the Seiberg-Witten effect alone is responsible.]

We begin by generalizing the analysis of Hiscock and Weems [Henceforth, HW] to charged AdS toral and planar black holes, in order to substantiate our claim that the temperatures of these black holes do indeed drop when they radiate.

\addtocounter{section}{1}
\section* {\large{\textsf{2. Evaporating Charged AdS Black Holes}}}
Four-dimensional\footnote{We work in four dimensions only for the sake of simplicity; we expect the same qualitative results to hold in higher dimensions.} AdS-Reissner-Nordstr$\ddot{\m{o}}$m black holes with \emph{flat} event horizons [henceforth, ``charged flat black holes''] have metrics of the form [see \cite{kn:77}]
\begin{equation}\label{ARRGH}
g(\text{FAdSRN}) = -\, \Bigg[{r^2\over L^2}\;-\;{8 \pi M^*\over r}+{4\pi Q^{*2}\over r^2}\Bigg]dt^2\; + \;{dr^2\over {r^2\over L^2}\;-\;{8 \pi M^*\over r}+{4\pi Q^{*2}\over r^2}} \;+\;r^2\Big[d\psi^2\;+\;d\zeta^2\Big],
\end{equation}
where $\psi$ and $\zeta$ are dimensionless coordinates on a flat space, and where the mass and charge parameters $M^*$, $Q^*$, are defined as follows. In the case in which the event horizon is compact, we shall take it to be a flat square torus with area 4$\pi^2 K^2$, where $K$ is a dimensionless ``compactification parameter''. Then $M^*$ is defined as $M/(4\pi^2 K^2)$, and similarly $Q^* = Q/(4\pi^2 K^2)$, where $M$ and $Q$ are the physical mass and charge of the hole. If we wish to consider a non-compact [planar] event horizon, then we let $M$, $Q$, and $K$ tend to infinity in such a manner that $M^*$ and $Q^*$ remain finite. The densities of the mass and electric charge at the event horizon of the hole are then given, for both toral and planar cases, by $M^*/r^2_h$ and  $Q^*/r^2_h$, where $r = r_h$ locates the event horizon; note that $r_h$ can be computed if $M^*$ and $Q^*$ are given.

In a holographic approach, $M^*$ and $Q^*$ are fixed by physical properties of the dual field theory, namely its energy density and chemical potential. [The formula for the electromagnetic potential also involves $Q^*$ rather than $Q$.] Indeed, one could \emph{define} $M^*$ and $Q^*$ in that way. Similarly, the time coordinate $t$ in the above formula for $g(\text{FAdSRN})$, which does not have a simple interpretation in the bulk, can be defined as proper time at infinity [where the metric is locally Minkowskian]. Henceforth, all of our references to ``rates of change'' will implicitly involve this proper time at conformal infinity.

We will assume that the usual conditions for holography to apply will always hold: that the string coupling and the ratio of the string length scale to the AdS curvature scale $L$ are small. In particular, one should think of $L$ as ``large''. Now one can compute the Kretschmann scalar [the square of the curvature tensor] for $g(\text{FAdSRN})$: it is given by
\begin{equation}\label{KRET}
R^{abcd}R_{abcd}(\text{FAdSRN})\; = \;{8\left (96\pi^2L^4M^{*2}r^2\,-\,192\pi^2L^4M^*Q^{*2}r\,+\,112\pi^2L^4Q^{*4}\,+\,3r^8\right)\over r^8L^4}.
\end{equation}
The maximal squared curvature for any point not inside the event horizon is of course attained at the event horizon. For very cold [nearly extremal] black holes of this kind [the condition for extremality being $Q^{*6} = (27/4)\pi M^{*4}L^2$], one finds that the squared curvature takes a remarkable form:
\begin{equation}\label{KRETEVENT}
R^{abcd}R_{abcd}(\text{FAdSRN; Extremal;} \;r = r_h)\; = \; {144\over L^4}.
\end{equation}
That is, since $L$ is assumed to be ``large'', the spacetime curvature outside a cold black hole of this sort is \emph{always very small}, independent of any other parameter. Whatever happens to the event horizon of such a black hole happens in the low-curvature regime.

Despite being the arena in which quantum gravity is best understood, asymptotically AdS spacetimes do not straightforwardly allow one to study Hawking evaporation of black holes ---$\,$ ``large'' asymptotically AdS black holes with spherical event horizons, and all planar and toral AdS black holes, tend ultimately to reach thermal equilibrium with their Hawking radiation. Nevertheless, large black holes can be made to evaporate by coupling the boundary field theory with an auxiliary system, such as another CFT \cite{kn:apologia, kn:rocha, kn:Raamsdonk} [or by attaching a Minkowski space to an AdS throat geometry \cite{kn:HH}]. Admittedly, this is a somewhat dubious procedure, since the ``CFT - AUX'' system is not well-understood, especially in the non-equilibrium context where we need to use it.
Ultimately, this problem will probably only be resolved by a fully dynamical analysis, of the type reviewed in \cite{kn:chesler}. In this work, we simply assume that some mechanism of this kind\footnote{These black holes can also evaporate if one \emph{artificially} ``mines'' the black holes, an operation that overcomes the effective potential around the holes. This is discussed in, for example \cite{kn:HH}; but see \cite{Brown} for a discussion of the subtleties of such operations.} can be made to work, and investigate the consequences, following Hiscock and Weems [HW] [who of course dealt only with the asymptotically flat case].

Following HW, we will work in the so-called ``relativistic units'' \cite{Israel} in which both the speed of light $c$ and Newton's constant $G$ are unity but the reduced Planck's constant $\hbar$ is not. Consequently, $\hbar G/c^3 = \hbar \approx 3 \times 10^{-66} ~\text{cm}^2$. This means that, unlike the usual convention in which $\hbar$ is set to unity and temperature has dimension of inverse length, in our choice of units temperature has dimension of length.
The radiation constant\footnote{Note that HW simply refer to the radiation constant as Stefan-Boltzmann constant.}, which is $4/c$ times the Stefan-Boltzmann constant, is denoted by $a=\pi^2/15\hbar^3$, with Boltzmann constant $k_B=1$. Without loss of generality, we will choose the charge of the black hole to be positive. However, we will use the Lorentz-Heaviside units, in which a factor of $4\pi$ appears in the Coulomb's Law but not in Maxwell's equations. Therefore, $Q^2$ in HW will appear as $Q^2/4\pi$ in our work.  The electron charge will be $e/\sqrt{4\pi}$ and its mass $m = 10^{-21} e/\sqrt{4\pi}$, where
$e= 6 \times 10^{-34} ~\text{cm}$. In addition, $Q_0 := \hbar e/\pi m^2 \approx 3.18 \times 10^{10} ~\text{cm}$ is the inverse of the Schwinger's critical field $E_c :=\pi m^2 c^3/(\hbar e)$.

Let us start with a short review of the HW analysis.

Since the problem of charged black hole evaporation is rather complicated, HW's analysis is restricted to the case in which the black holes are cold. In the asymptotically flat case, this means that the black hole is necessarily large. Due to the low temperature, HW can reasonably assume that only massless particles are emitted via thermal Hawking evaporation, and treat charge loss as the result of Schwinger effect \cite{Schwinger}. [In fact, a result due to Gibbons \cite{Gibbons} is that, as long as the black hole is much larger in radius than the reduced Compton wavelength of the electron, that is, $M \gg \hbar/m \approx 10^{-10}$ cm $\approx 10^{18}$ g, then the pair-production of charged particles is well-approximated by flat-space quantum electrodynamics [QED]. Intuitively, for large enough mass, the curvature radius of the 2-sphere is larger than the size of an electron.]

Note that, as HW emphasise, although the production of charged particles are treated separately from the thermal Hawking flux of neutral particles in this model, they are actually all part of Hawking emission. In other words, the charged particle emission is actually thermodynamically related to a non-zero chemical potential associated with the electromagnetic field of the black hole. The effective decoupling between thermal emission of neutral particles and electromagnetic [as opposed to gravitational] creation of charged particles is due to the \emph{low temperature} of large black holes \cite{Gibbons}, although it has been argued that the Schwinger mechanism and the Hawking radiation are generically indistinguishable for near-extremal black holes \cite{CM}.

For simplicity, HW assumed that the electromagnetic field is weak enough that we may ignore the contributions of muons and other heavier charged particles, and only deal with electrons and positrons. Schwinger's formula, which describes the rate of pair creation per unit 4-volume $\Gamma$, is, in the case applicable to the HW analysis,
\begin{equation}
\Gamma = \frac{e^2}{16\pi^4 \hbar^2}\frac{Q^2}{4\pi r^4} \text{exp}\left(-\frac{4\pi^2 m^2 r^2}{\hbar e Q}\right) \times \left[1 + O\left(\frac{e^3Q}{16\pi^2 m^2 r^2}\right)\right].
\end{equation}
The ``weak-field approximation'' means that one ignores all higher order terms, which is valid provided that
\begin{equation}\label{weakfield1}
\frac{e^3 Q}{16\pi^2 m^2 r^2} \ll 1, ~~\text{for all}~ r \geq r_h,
\end{equation}
where $r_h$ locates the event horizon of the black hole.

For an asymptotically flat Reissner-Nordstr\"om black hole with mass $M$ and charge $Q$,
\begin{equation}
r_h [\text{RN}] = M + \sqrt{M^2 - \frac{Q^2}{4\pi}},
\end{equation}
and thus
\begin{equation}
M \gg \frac{e^3}{16\pi^2 m^2} \approx 3.8 \times 10^6 ~\text{cm}.
\end{equation}
That is, the black hole has to be large [and therefore cold] enough to satisfy this.

The charge loss rate is thus given by the integral
\begin{flalign}
\frac{dQ}{dt} &\approx \frac{e^3}{4\pi^3 \hbar^2} \int_{r_h}^\infty \frac{Q^2}{4\pi r^2} \exp{\left(-\frac{4\pi r^2}{Q_0Q}\right)} dr\\
& =\frac{e^3}{8\pi^{7/2} \hbar^2} \left[-\frac{Q^{3/2}}{2\sqrt{Q_0}} ~\text{erf}\left(\frac{4\pi r}{\sqrt{Q_0 Q}} \right)- \frac{Q^2}{4\pi r} \exp{\left(-\frac{4\pi r^2}{Q_0 Q}\right)}\right]\Bigg|_{r_h}^\infty.
\end{flalign}
For a sufficiently large black hole, HW can use the series approximation for the complementary error function $\text{erfc}(x) = 1- \text{erf}(x)$:
\begin{equation}\label{series}
\text{erfc}(x) = \frac{e^{-x^2}}{x\sqrt{\pi}}\left[1 + \sum_{n=1}^{\infty}  (-1)^n \frac{1 \cdot 3 \cdot 5 \cdots (2n-1)}{(2x^2)^n} \right], ~~x \gg 1,
\end{equation}
which then yields, finally, the charge loss rate
\begin{equation}
\frac{dQ}{dt} \approx - \frac{e^4}{2^8\pi^{13/2} \hbar m^2} \frac{Q^3}{r_h^3} \exp{\left(-\frac{4\pi r_h^2}{Q_0Q}\right)}.
\end{equation}

On the other hand, the mass loss of the black hole is given by
\begin{equation}\label{massloss0}
\frac{dM}{dt} = -\alpha aT^4\sigma + \frac{Q}{r_h}\frac{dQ}{dt},
\end{equation}
where the first term on the right describes thermal mass loss due to Hawking radiation\footnote{If black hole information is indeed not lost, the final Hawking radiation has a thermal \emph{spectrum} despite being in a pure state [that is, there is high degree of entanglement] instead of in a thermal \emph{state}. Of course, in many models of Hawking radiation, excitations over the thermal spectrum are to be expected. Nevertheless, the spectrum is still very close to thermal, thus justifying the use of these equations. For this distinction see \cite{kn:SH2}; see also Section V of \cite{MP} for a related discussion.}. This is just the Stefan-Boltzmann law ---$\,$ the power emitted by a black body of temperature $T$ is
\begin{equation}\label{SB}
P = \frac{a}{4}\,\alpha  \times \text{Area} \times T^4 = \frac{a}{4}\,\alpha \times 4\sigma \times T^4,
\end{equation}
where $\sigma$ is the cross section of the black body in the case of spherical symmetry, and $\alpha$ is another constant to which we shall shortly return. In the case of asymptotically flat black hole spacetimes, $\sigma$ is the geometric cross section, which is related to the innermost [unstable] photon orbit of the black hole spacetime. For asymptotically flat Reissner-Nordstr\"om geometry, we have \cite{kn:HW}
\begin{equation}
\sigma = \frac{\pi}{8} \left(3M + \sqrt{9M^2 - \dfrac{2Q^2}{\pi}}\right)^4 \left(3M^2 - \dfrac{Q^2}{2\pi} + M\sqrt{9M^2 - \dfrac{2Q^2}{\pi}}\right)^{-1},
\end{equation}
which recovers the familiar expression for the geometric optics cross section of Schwarzschild geometry [see, e.g., equation (6.3.34) of \cite{Wald}], $27\pi M^2$, in the $Q \to 0$ limit.

In addition, the mass loss due to electromagnetic pair creation is described by a term proportional to the pair-creation rate, and to the electromagnetic potential energy lost per pair created\footnote{The second term of equation (\ref{massloss0}) is just the term that appears in the first law of black hole mechanics: \newline $ dM=\dfrac{\kappa}{8\pi} dA + \dfrac{Q}{r_h} dQ + \Omega dJ$.}. The constant $\alpha$ mentioned above is a quantity of order unity that depends on how many species of massless particles are present [essentially the so-called ``greybody factor'']. HW showed that the qualitative -- and, to a large degree, the quantitative -- results are not sensitive to the exact value of $\alpha$.

More precisely, HW consider the possible number of massless neutrino species $n_\nu=0,1,2,3$. Each choice gives rise to a corresponding value of $\alpha$. The different $\alpha$'s contribute an $O(1)$ difference to the lifetime of the asymptotically flat Reissner-Nordstr\"om black hole. Admittedly, in a \emph{``stringy''} AdS bulk there would be other massless particles beyond the standard model of particle physics. Nevertheless, as we will see, the time scale involved is so enormously large that an $O(1)$ or even $O(1000)$ difference would not change the result appreciably. We will henceforth set $\alpha=1$ for simplicity. Note that if one indeed considers massless particle species in addition to the photon and graviton, then the lifetime of the black hole will in fact be \emph{shortened} [more energy radiated thermally per unit time], and this would favour the HH proposal.

This model of evaporating asymptotically flat charged black holes can be generalized to asymptotically locally AdS black holes, and in particular to those black holes with flat horizons. In the asymptotically flat case, we saw that, in order for flat-space QED to be applicable, one needs a sufficiently large black hole to ensure that the curvature radius of the underlying spherical geometry is larger than the size of an electron. For flat black holes it would seem that this condition is automatically satisfied, the curvature radius being infinite. However, there are a few subtleties here. For simplicity, let us first consider electrically neutral toral black holes. The Hawking temperature is [see, for example, \cite{Birmingham, KlemmVanzo}]
\begin{equation}
T[Q=0] = \frac{3\hbar r_h}{4\pi L^2}.
\end{equation}
We see that, unlike its asymptotically flat cousin, a toral black hole has a temperature proportional to its radius. Thus, for any fixed compactification parameter $K$, a larger black hole is hotter. This is of course related to the fact that these black holes have \emph{positive} specific heat, unlike the Schwarzschild black hole.

Since our aim is to study cold black holes, and also to use the method of HW, in which thermal mass loss can be cleanly separated from charge loss, we need to make sure that our black holes are not too hot. For a neutral black hole, this means that we want
\begin{equation}
T[Q=0] = \frac{3\hbar r_h}{4\pi L^2} < 2m.
\end{equation}
Since the event horizon for neutral toral black hole is located at a value of $r$ given by
\begin{equation}
r_h=\left(\frac{2ML^2}{\pi K^2}\right)^{\frac{1}{3}},
\end{equation}
the inequality translates to an upper bound on $M$, given by
\begin{equation}
M < 2^8 \pi^4K^2L^4 \left(\frac{m}{3\hbar}\right)^3.
\end{equation}
For $L=10^{15}$ cm, say, we get roughly
\begin{equation}\label{massbound}
M < 1.12 \times K^2 \times 10^{97} ~\text{cm}.
\end{equation}
Of course, a charged black hole will have colder temperature, and therefore, can tolerate higher upper bound on the mass without emitting charged particles thermally. Nevertheless, for convenience, we will always choose the initial condition for mass to be below the bound given in equation (\ref{massbound}).

Next, we need to find the circumstances under which the weak-field condition for the Schwinger effect holds. [We remind the reader that the weak-field requirement allows us to consider only positrons and electrons, not charged particles of higher mass like muons. This is reasonable since the pair-creation rate depends exponentially on the
square of the mass of the particle species.] Recall that it is the electric field strength $E=Q^*/r^2$ that is important in the pair-creation of charged particles, not the charge $Q$ \emph{per se}. In terms of electric field strength, the Schwinger formula [equation (\ref{weakfield1})] is
\begin{equation}
\Gamma = \frac{e^2}{16 \pi^4 \hbar^2} E^2 \exp{\left(-\frac{\pi m^2 \sqrt{4\pi}}{\hbar e E}\right)} \times \left[1 + O\left(\frac{e^3E}{m^2(4\pi)^{3/2}}\right)\right].
\end{equation}
For our toral geometry, this expression yields
\begin{equation}\label{weakfield2}
\Gamma = \frac{e^2Q^2}{256 \pi^8 \hbar^2 K^4 r^4} \exp{\left(-\frac{8m^2\pi^{7/2}K^2r^2}{\hbar e Q}\right)} \times \left[1 + O\left(\frac{e^3 Q}{2^5 \pi^{7/2} m^2 K^2 r^2}\right)\right].
\end{equation}
The dependence on $K^2$ in the exponential term, which dominates the Schwinger effect, is of course natural -- due to the conservation of flux, for any fixed charge, one expects a black hole with large area [that is, large $K$] to have a weaker field.

This implies that for the ``weak-field'' approximation to hold, we need
\begin{equation}
\frac{e^3}{32 \pi^{7/2} K^2 m^2} \ll \inf{\left(\frac{r^2}{Q}\right)} = \frac{\inf{(r^2)}}{\sup{(Q)}} = \frac{r^2_{\text{ext}}}{Q_{\text{ext}}} = \frac{\left(\frac{ML^2}{2\pi K^2}\right)^{\frac{2}{3}}}{(108 \pi^5 M^4 L^2 K^4)^{1/6}}=\frac{L}{2\cdot 3^{1/2} \pi^{3/2} K^2},
\end{equation}
where
\begin{equation}\label{Qext}
Q_{\text{ext}}= (108 \pi^5 M^4 L^2 K^4)^{1/6}
\end{equation}
is the extremal charge, and
\begin{equation}\label{rext}
r_{\text{ext}}=\left(\frac{ML^2}{2\pi K^2}\right)^{\frac{1}{3}}
\end{equation}
locates the event horizon of the extremal black hole.

To summarise, we have the following result.
\begin{proposition}
The weak-field condition for the validity of the Schwinger formula in the case of asymptotically locally AdS black holes with flat event horizons in $(3+1)$-dimensions is
\begin{equation}
\frac{e^3}{m^2} \ll \frac{16\pi^2 L}{\sqrt{3}},
\end{equation}
that is,
\begin{equation}
L \gg 6.6 \times 10^6 \text{cm},
\end{equation}
independent of the mass of the black hole.
\end{proposition}
Unlike the asymptotically flat case then, the AdS case requires us to consider large $L$, that is, \emph{small cosmological constant}, not large $M$. In fact, as we have seen, $M$ is bounded above; and in addition, as we shall show later, phase transitions also put constraints on the value the mass can take.
In addition, with the expressions for the extremal charge and the extremal horizon, that is, equations (\ref{Qext}) and (\ref{rext}), the requirement that the series approximation [equation (\ref{series})] is valid yields, for charged flat black holes, $L \gg 1.21 \times 10^8$ cm. This clearly also satisfies the inequality obtained above in Proposition 1.
Henceforth, in our numerical analysis, we shall fix $L=10^{15}$ cm for definiteness. We will discuss the effect of varying $L$ in Section 4.

The fact that the weak-field condition is independent of the black hole mass is interesting in its own right. In fact, in some sense, these toral black holes behave more like empty AdS than like asymptotically flat black holes. A simple example of this is given by calculating the maximal in-falling time from horizon to singularity for a neutral toral black hole. In the Schwarzschild case, we have
\begin{equation}
\tau_{\text{max}} = \int_0^{2M} \left(\frac{2M}{r}-1\right)^{-\frac{1}{2}} dr = \pi M,
\end{equation}
but, for a neutral toral black hole, we have instead
\begin{equation}
\tau_{\text{max}} = \int_0^{r_h} \left(\frac{2M}{\pi K^2 r}-\frac{r^2}{L^2}\right)^{-\frac{1}{2}} dr = \frac{\pi L}{3}, ~~ r_h = \sqrt[3]{\frac{2ML^2}{\pi K^2}},
\end{equation}
which is again independent of the black hole mass. This is reminiscent of the fact that the time to fall from anywhere to the ``centre'' of AdS only depends on the curvature radius. A similar observation holds in relation to the geometric cross section $\sigma$, to which we now turn.

It turns out that the usual definition of the geometric cross section for asymptotically flat black holes does \emph{not} carry over straightforwardly to the toral AdS case. Recall that the geometric cross section is by definition $\sigma=\pi b^2$, where $b$ is the maximum impact parameter for a massless particle to be captured. The computation of the impact parameter in the asymptotically flat case normally proceeds by normalizing the asymptotic energy of the particle as $\mathcal{E} \to 1$. In the asymptotically AdS case, however, $\mathcal{E} \to \infty$ toward the boundary.

Fortunately, this is misleading -- we need \emph{not} define $b$ at all for our purpose of studying emission of Hawking radiation. We are only interested in particles that can \emph{escape} the black hole to infinity, not be \emph{captured}. In the asymptotically flat case, these two notions are interchangeable since the photon orbit corresponds to the local maximum of the effective potential experienced by massless particles [see Figure (6.5) of \cite{Wald}]. However, for toral black holes, the potential reads
\begin{equation}
V[r] = \frac{J^2}{r^2}\left(\frac{r^2}{L^2} - \frac{8 M^*}{ r} + \frac{4\pi Q^{*2}}{r^2}\right),
\end{equation}
where $J$ is the angular momentum of the particle. This potential is monotonically increasing and approaches the asymptotic value $J^2/L^2$. Therefore, in our case, ``escape'' is not the same as ``capture'', indeed every ingoing massless particle reaches the black hole, but not all massless particles can escape.

Evidently, given a fixed angular momentum $J$, the particle needs to climb over the potential barrier of height $J^2/L^2$ to reach infinity. The metric, restricted on the equatorial plane, yields the equation of motion
\begin{equation}
-f(r)\left(\frac{dt}{d\lambda}\right)^2 + f(r)^{-1} \left(\frac{dr}{d\lambda}\right)^2 + r^2 \left(\frac{d\phi}{d\lambda}\right)^2 = 0,
\end{equation}
where $\lambda$ is a parameter for null geodesics, and where
\begin{equation}
f(r):=\frac{r^2}{L^2} - \frac{8 M^*}{ r} + \frac{4\pi Q^{*2}}{r^2}.
\end{equation}

We have
\begin{equation}
\mathcal{E}=f(r)\frac{dt}{d\lambda}, ~~ J=r^2\frac{d\phi}{d\lambda},
\end{equation}
where $\mathcal{E}$ is the energy the particle needs to arrive at $V_\infty = J^2/L^2$. That is, $\mathcal{E} \geq J^2/L^2$. At infinity we must have
\begin{equation}
\left(\frac{dr}{d\lambda}\right)^2 = \left[\left(\frac{J}{L}\right)^2 - \frac{J^2}{r^2}\right]f(r),
\end{equation}
which vanishes when $L=r$. One can then define the ``cross section'' $\sigma \propto L^2$, which is again independent of the black hole mass, as well as its charge. This simple expression for the cross-section agrees with the one given in \cite{KlemmVanzo}.

We are now in a position to generalize the HW analysis.

The area appearing in the Stefan-Boltzmann law in equation (\ref{SB}) is now $4\pi^2 K^2 L^2$ and so the differential equation governing mass loss is
\begin{equation}\label{massloss}
\frac{dM}{dt} = - a\pi^2 K^2 L^2 T^4 + \frac{Q}{4\pi^2 K^2 r_h}\frac{dQ}{dt},
\end{equation}
where the Hawking temperature is
\begin{equation}\label{Hawking}
T = \frac{\hbar}{2\pi^2K^2}\left[\frac{1}{r_h^2}\left(3M - \frac{Q^2}{2\pi^2 K^2 r_h}\right)\right],
\end{equation}
or, in terms of AdS length scale,
\begin{equation}
T = \hbar\left[\frac{r_h}{\pi L^2} - \frac{M}{2\pi^2 K^2 r_h^2}\right].
\end{equation}

The differential equation governing charge loss in the weak-field limit can be obtained by integrating the leading term of equation (\ref{weakfield2}); it is given by
\begin{equation}\label{chargeloss}
\frac{dQ}{dt} \approx - \frac{e^4 K^2}{1024 \pi^{19/2} \hbar m^2 K^6} \frac{Q^3}{r_h^3} \exp{\left(-\,\frac{8\pi^{7/2}K^2m^2r_h^2}{\hbar e Q}\right)}.
\end{equation}
In terms of $M^*$ and $Q^*$, these coupled ordinary differential equations read:
\begin{equation}
\begin{cases}
\dfrac{dM^*}{dt} = - \dfrac{a}{4}L^2 T^4 + \dfrac{Q^*}{r_h}\dfrac{dQ^*}{dt},
\\
\\
\dfrac{dQ^*}{dt} \approx - \dfrac{e^4}{64 \pi^{11/2} \hbar m^2} \dfrac{Q^{*3}}{r_h^3} \exp{\left(-\,\dfrac{2\pi^{3/2}m^2r_h^2}{\hbar e Q^*}\right)},
\end{cases}
\end{equation}
where the Hawking temperature is
\begin{equation}
T = \frac{\hbar}{r_h^2}\left[6M^* - \frac{4Q^{*2}}{r_h} \right] = \hbar\left[\frac{r_h}{\pi L^2} - \frac{2M^*}{r_h^2}\right].
\end{equation}
These expressions also hold in the case of planar black hole.

\addtocounter{section}{1}
\section* {\large{\textsf{3. Thermodynamics of Charged Evaporating Flat Black Holes}}}

We first note that in the case of neutral evaporating toral black holes, the rate of mass loss is
\begin{equation}
\frac{dM}{dt} = -a\pi^2 K^2 L^2 T^4 = -a \pi^2 K^2 L^2 \left[\frac{\hbar}{2\pi^2 K^2}\frac{3M}{r_h^2}\right]^4,
\end{equation}
where the event horizon is located at
\begin{equation}
r_h = \left(\frac{2ML^2}{\pi K^2}\right)^{\frac{1}{3}}.
\end{equation}
Therefore $dM/dt \propto -\, M^{4/3}$, which implies that $M(t)$ only reaches zero asymptotically. This is in contrast to the case of a Schwarzschild black hole, for which, as is well known, zero mass is attained in a finite time [though, again, the spacetime curvature near such a black hole eventually becomes so large that we have no good reason to trust semi-classical physics in the final stages of its evaporation]. It is noteworthy that even uncharged toral black holes already threaten the HH proposal [as HH themselves point out]. We will return to this point later.

It is well known that electrically neutral, [quasi-]static flat AdS black holes have a positive specific heat. However, in our set-up, in which charged flat AdS black holes are allowed to evaporate, it is \emph{a priori} possible that the specific heat can change sign at some point in the evolution of the black hole, just as Hiscock and Weems found in the case of the evaporating asymptotically flat Reissner-Nordstr\"om black hole. It is therefore important to check the specific heat of these black holes. We emphasise that, on physical grounds, one should \emph{not} hold the charge fixed when calculating the specific heat [though it can be instructive to see what happens if that is done, see below]; instead one should directly compute it using
\begin{equation}
C:=\frac{dM}{dT} = \frac{dM}{dt}\left(\frac{dT}{dt}\right)^{-1},
\end{equation}
as HW did. Now note that $dM/dt$ is always negative. Thus, the sign of the specific heat is the opposite of the sign of $dT/dt$.

For any fixed compactification parameter $K$, we shall prove that, as one would expect\footnote{This still needs to be checked explicitly since it is \emph{possible} that the horizon area is \emph{not} monotonically decreasing. In fact, for some initial conditions, [asymptotically flat] Kerr black holes lose angular momentum much more rapidly than mass, resulting in their horizon area initially \emph{increasing} as they evaporate \cite{kn:page}.}, the black hole gets smaller as it evaporates [the same proof, \emph{mutatis mutandis}, also holds for charged planar AdS black holes, as well as asymptotically flat Reissner-Nordstr\"om black holes]:
\begin{proposition}
The value of the radial coordinate at event horizon, $r_h(t)$, is a monotonically decreasing function of time.
\end{proposition}
\begin{proof}
The defining equation of the event horizon is, from equation (\ref{ARRGH}),
\begin{equation}
0 = \frac{r_h^2}{L^2} - \frac{2M}{\pi K^2 r_h} + \frac{Q^2}{4\pi^3 K^4 r_h^2}.
\end{equation}
Taking the derivative with respect to $t$, we obtain
\begin{equation}
0 = \left(\frac{2r_h}{L^2} + \frac{2M}{\pi K^2 r_h^2} - \frac{Q^2}{2\pi^3K^4r_h^3}\right)\frac{dr_h}{dt} - \frac{2}{\pi K^2 r_h}\frac{dM}{dt} + \frac{Q}{2\pi^3 K^4 r_h^2}\frac{dQ}{dt}.
\end{equation}
The expression in the brackets is just $4\pi/\hbar$ times the Hawking temperature, and so
\begin{equation}\label{TMQrelation}
\frac{4\pi T}{\hbar}\frac{dr_h}{dt} = \frac{2}{\pi K^2 r_h}\frac{dM}{dt} - \frac{Q}{2\pi^3 K^4 r_h^2}\frac{dQ}{dt}.
\end{equation}
Upon substituting this into the mass loss equation, equation (\ref{massloss}), we find that the $dQ/dt$ term cancels [of course $dr_h/dt$ still implicitly depends on charge loss rate via $T=T(M,Q)$], and we are left with:
\begin{equation}
\frac{4\pi T}{\hbar}\frac{dr_h}{dt} = -\frac{2a\pi L^2 T^4}{r_h} \leq 0,
\end{equation}
with equality attained only in the extremal case, at which $T=0$.
\begin{flushright}\qed
\end{flushright}

\end{proof}

We may describe the evolution of the generic horizon by means of a dimensionless function $\gamma(t)$, defined by
\begin{equation}\label{r}
r_h^3(t) = \frac{\gamma(t)M(t)L^2}{\pi K^2},
\end{equation}
where $\gamma(t) \in [1/2, 2]$ is \emph{not} necessarily monotonically decreasing.
The case $\gamma=2$ corresponds to a neutral black hole, while $\gamma=1/2$ describes an extremal black hole. Note that, due to the competition between $\gamma(t)$ and $M(t)$, we cannot decide, by appealing to the monotonicity of $r_h$ alone, whether the black hole will evolve towards the extremal limit or towards the zero-mass limit. [In principle, one can solve for $r_h$ explicitly from the metric, but the expression is too complicated to be of practical use for analytic calculations.]

From the expression for the Hawking temperature in equation (\ref{Hawking}), we can now compute its time derivative:
\begin{equation}\label{T1}
\frac{dT}{dt} = \hbar \left[\left(\frac{1}{\pi L^2} + \frac{M}{\pi^2 K^2 r_h^3}\right)\frac{dr_h}{dt} - \frac{1}{2\pi^2 K^2 r_h^2}\frac{dM}{dt}\right],
\end{equation}
where
\begin{equation}\label{T2}
\frac{dr_h}{dt} = \frac{1}{3}\left(\frac{\gamma ML^2}{\pi K^2}\right)^{-\frac{2}{3}}\left(\frac{\gamma L^2}{\pi K^2}\frac{dM}{dt} + \frac{ML^2}{\pi K^2}\frac{d\gamma}{dt}\right).
\end{equation}
We can now compute the specific heat.
First we note that the expression
\begin{equation}
\frac{1}{3}\left(\frac{1}{\pi L^2} + \frac{M}{\pi^2 K^2 r_h^3}\right)\left(\frac{\gamma ML^2}{\pi K^2}\right)^{-\frac{2}{3}}\frac{\gamma L^2}{\pi K^2} - \frac{1}{2\pi^2K^2r_h^2}
\end{equation}
can be simplified to
\begin{equation}
\left[\frac{1}{3}\left(1+\gamma\right)-\frac{1}{2}\right]\left(\gamma \pi^2 KL^2 M\right)^{-2/3}.
\end{equation}
Since $\gamma \in [1/2, 2]$, this expression is always positive except for the extremal case in which the expression is identically zero. Thus
\begin{flalign}
\frac{dT}{dt} =& ~\hbar \left[\frac{1}{3}\left(1+\gamma\right)-\frac{1}{2}\right]\left(\gamma \pi^2 KL^2 M\right)^{-2/3} \frac{dM}{dt} \\ \notag &+ \frac{\hbar}{3}\left(\frac{1}{\pi L^2} + \frac{M}{\pi^2 K^2 r_h^3}\right)\left(\frac{M L^2}{\pi K^2}\right)^{\frac{1}{3}}\gamma^{-\frac{2}{3}}\frac{d\gamma}{dt},
\end{flalign}
in which the first term is negative, due to the fact that $dM/dt < 0$.
Now, $\text{sgn} (C) = -\text{sgn} (dT/dt)$. The specific heat is therefore positive only if the contribution from $d\gamma/dt$ term never becomes too positive.
Thus, indeed we cannot conclude that the black hole always has positive specific heat \emph{a priori}. Nevertheless, our numerical results, for example, the left plot of Figure (\ref{2}), do suggest that $dT/dt$ is always negative, and thus that the specific heat is always positive for evaporating charged flat black holes. In fact, the numerical results suggest that $\gamma(t)$, far from becoming too large, is in fact monotonically decreasing. [On the other hand, for some \emph{asymptotically flat} Reissner-Nordstr\"om black holes, $\gamma(t)$ does eventually change sign.] See the right plot of Figure (\ref{2}). We remark that the same result holds in the planar case.
\begin{figure}[!h]
\centering
\mbox{\subfigure{\includegraphics[width=3in]{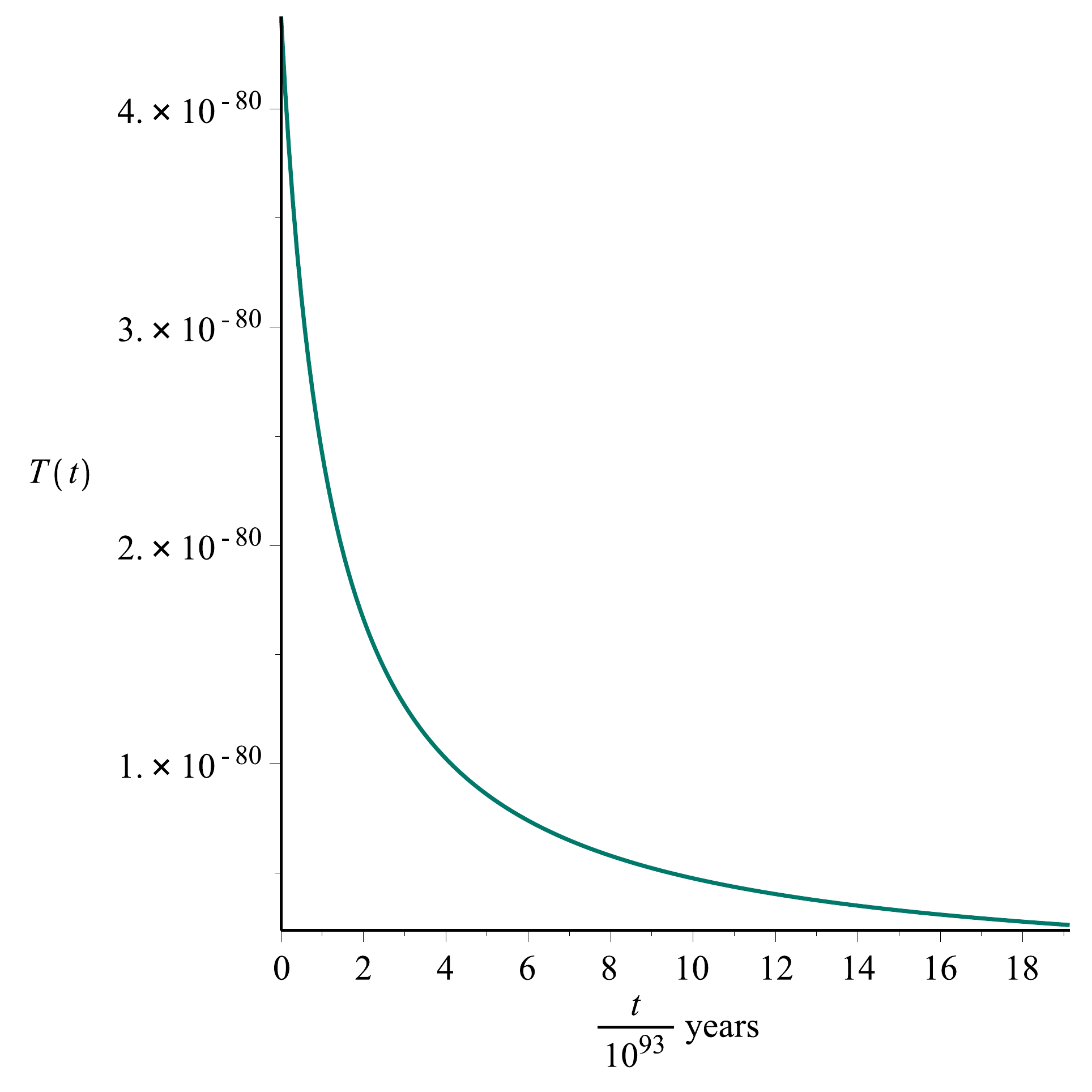}}\quad
\subfigure{\includegraphics[width=3in]{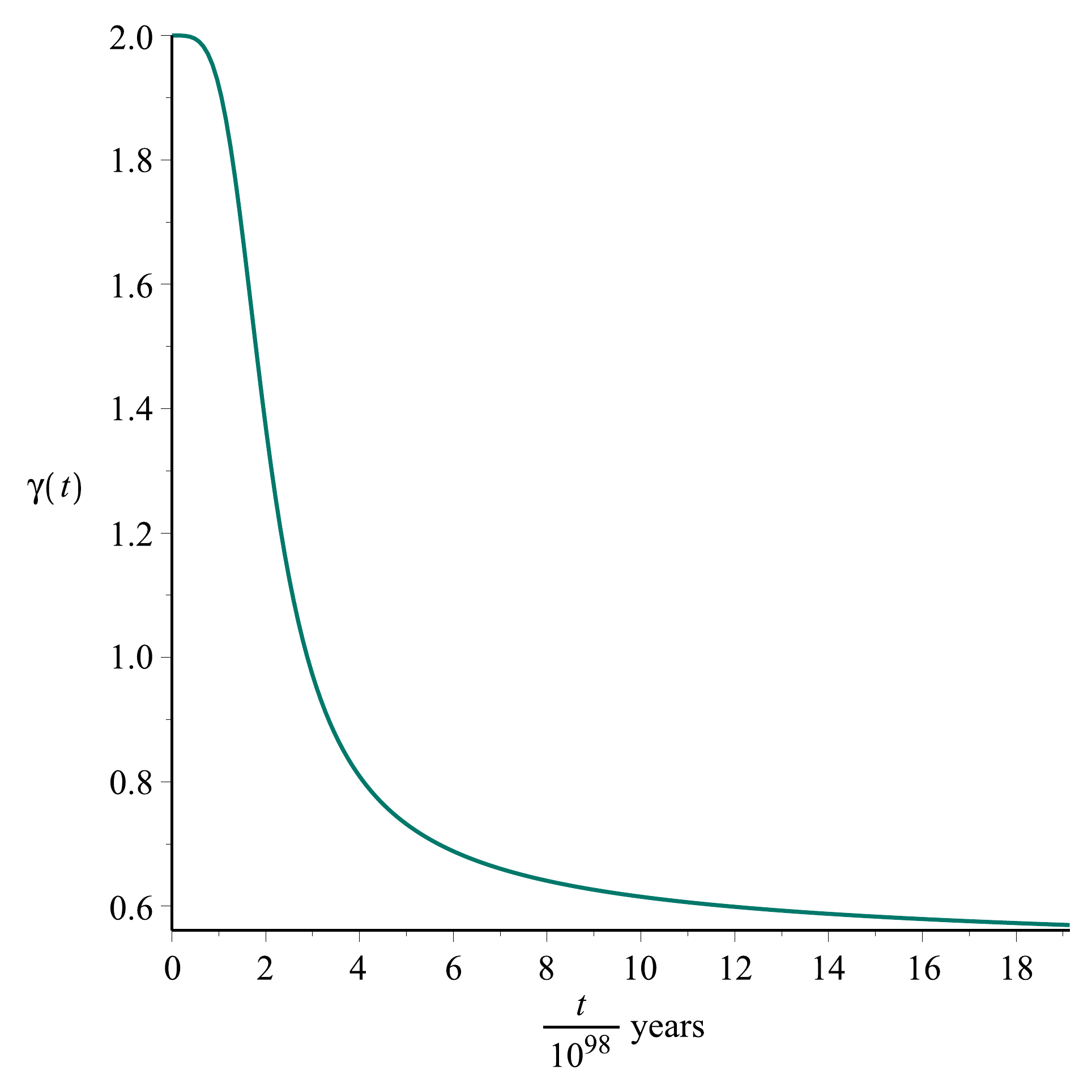} }}
\caption{\textbf{Left:} The temperature [units of centimetres] as a function of time, of a charged toral black hole with $K=1$, and initial condition $M(0)=5.6 \times 10^{20}$ cm, $Q(0)=1.7 \times 10^{9}$ cm. The initial temperature is evaluated to be about $4.42 \times 10^{-80} $ cm. \textbf{Right:} The [dimensionless] function $\gamma(t)$ of the same black hole is monotonically decreasing. Note that $\gamma(0)$ is extremely close [since the black hole is initially very close to extremal limit], but not exactly equal, to 2. \label{2}}
\end{figure}

Although, as we mentioned, on physical grounds we should not hold the charge fixed when calculating the specific heat, it is instructive to do precisely this. For, in the large mass limit, HW recover the classic result of Davies \cite{Davies}: by holding the charge fixed, one finds that sufficiently highly charged asymptotically flat Reissner-Nordstr\"om black holes have positive specific heat. In other words, the $Q=\text{const.}$ case allows us to probe certain limits of the parameter space. In fact, our numerical results in the next section show that charged flat black holes \emph{do} maintain $Q \approx \text{const.}$ along their evolutionary history, contrary to asymptotically flat spacetime intuitions\footnote{The details of the underlying physics of charge dissipation for these black holes is discussed in \cite{kn:elsewhere}.}.

For $Q = \text{const.}$, from Eq.(\ref{TMQrelation}), we have
\begin{equation}\label{TMrelation}
\frac{dr_h}{dt} = \frac{\hbar}{2\pi^2 T K^2 r_h}\frac{dM}{dt}.
\end{equation}
Substituting this expression into Eq.(\ref{T1}) and simplifying, we obtain
\begin{equation}
\frac{dT}{dt} = \hbar\left[\frac{3}{2\gamma(t)-1}\right]\frac{1}{2\pi^2K^2r_h^2}\frac{dM}{dt}.
\end{equation}
Since $\gamma(t) \in [1/2, 2]$, and $dM/dt < 0$, we see that $dT/dt$ is always negative and diverges to $-\infty$ as extremality is approached. Consequently, the specific heat is always positive [and tends to zero in the extremal limit] if we hold the electric charge fixed. The results can be appreciated from the plot of temperature as a function of $M$ and $Q$, as depicted in Figure (\ref{2(2)}). First recall that holding charge fixed means that
\begin{equation}
\frac{dT}{dt}=\frac{\partial T}{\partial M} \frac{dM}{dt},
\end{equation}
so $\text{sgn}(C)=\text{sgn}(\partial T/\partial M)$.
For some asymptotically flat Reissner-Nordstr\"om black holes, there are regions in the parameter space [where the charge is sufficiently large] in which $\partial T/\partial M$ does become positive. However, charged flat black holes do not behave in that manner.

\begin{figure}[!h]
\centering
\mbox{\subfigure{\includegraphics[width=3in]{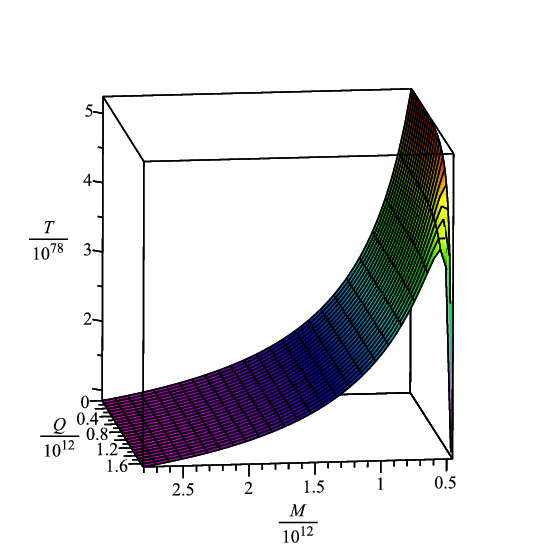}}\quad
\subfigure{\includegraphics[width=3in]{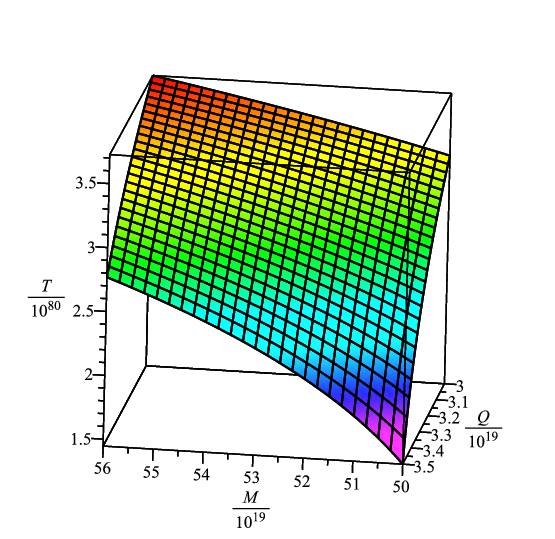} }}
\caption{\textbf{Left:} The temperature [units of centimetres] as a function of mass and charge for an asymptotically flat Reissner-Nordstr\"om black hole. \textbf{Right:} The temperature as a function of mass and charge for an AdS-Reissner-Nordstr\"om black hole with toral topology and $K=1$. \label{2(2)}}
\end{figure}

In the case of AdS black holes with spherical topology, Hawking and Page showed that cold black holes are not stable -- they undergo a phase transition into thermal AdS \cite{kn:hawpa}. For neutral toral black holes, it is known that a similar phase transition exists for cold black holes; however, the preferred state is not thermal AdS but a type of ``soliton'' \cite{kn:surya, kn:hormy}. The generalization to the charged case was considered in \cite{kn:AdSRN}, where the critical temperature below which the soliton configuration is thermodynamically preferred is found to be governed by the compactification parameter $K$:
\begin{equation}\label{PHASETEMP}
T_c = \frac{\hbar}{2\pi KL}.
\end{equation}
As with many properties of toral black hole spacetimes, this critical temperature has the property that, for any fixed $K$, it only depends on AdS length scale $L$, and independent of the mass and charge of the black hole.

If the black hole is to exist, then, the Hawking temperature must satisfy $T \geq T_c$. Explicitly,
\begin{equation}
\hbar \left[\frac{r_h}{\pi L^2}-\frac{2M}{4\pi^2 K^2 r_h^2}\right] \geq \frac{\hbar}{2\pi KL}.
\end{equation}
With the horizon parametrized by $\gamma(t)$, this yields a lower bound on the black hole mass
\begin{equation}
M(t) \gg \frac{8\pi L \gamma(t)^2}{K (4\gamma(t) -2)^3} =: M_c(t),
\end{equation}
where we have expressed the time-dependence explicitly.

However, note that $M_c$ is unbounded above as the black hole tends to extremality, that is, as $\gamma \to 1/2$. Thus we see that, even if one starts with a black hole with arbitrarily large mass, \emph{if} the black hole evolves towards the extremal limit, then the black hole mass [which is monotonically decreasing] \emph{will} eventually drop below $M_c$.

This means that, \emph{if} the phase transition temperature is not zero, then the black hole will be destroyed by a phase transition [at some very \emph{low} temperature] in a finite time. This time will be very long by normal standards, especially for black holes with large values of the compactification parameter $K$. However, the entropy of these black holes is also very large [being related to $K^2$], and this means, if the Harlow-Hayden conjecture [to the effect that the decoding time is exponential in the entropy] is correct, that the decoding time in this case is even more enormous. In every case, then, the black hole suffers a phase transition in a time which is utterly negligible relative to the decoding time.

There is however a crucial exception to this statement: the case of \emph{planar} black holes, with non-compact event horizons. For these black holes ---$\,$ which are in fact the most important ones in applications ---$\,$ there is \emph{no} phase transition, as one sees from equation (\ref{PHASETEMP}). Thus we still have a very important class of flat black holes which apparently have arbitrarily long lifetimes. This loophole must be closed, for otherwise we would arrive at the bizarre conclusion that the HH argument can only be made to work if the event horizon is compactified. We now proceed to do that.

\addtocounter{section}{1}
\section* {\large{\textsf{4. Fatal Attraction Toward Extremality}}}

As mentioned in Section 1, what we need to complete the argument is to show that, in addition to phase transitions, highly charged flat black holes are vulnerable to the brane pair-production instability discovered by Seiberg and Witten \cite{kn:seiberg}. This effect destabilizes a four-dimensional flat black hole when the electric charge is around 92\% of the extremal charge, and it does so \emph{both} in the toral \emph{and} in the planar cases.

Since the extremal charge is $Q_{\text{ext}}=(108\pi^5M^4L^2K^4)^{1/6}$, it is convenient to define
\begin{equation}\label{w}
w[M]:=\frac{(108\pi^5L^2K^4)^{1/6}}{M^{1/3}},
\end{equation}
so that the normalized charge-to-mass ratio satisfies
\begin{equation}\label{QwM}
\frac{\widetilde{Q}}{M}:=\frac{Q}{wM} \in [0,1].
\end{equation}
That is, the extremal case has $\widetilde{Q}/M = 1$.

The evolutionary history of charged evaporating flat black holes is easy to describe. Our numerical evidence indicates that, independent of the initial conditions,
they all evolve toward extremality, i.e. extremal limit is an \emph{attractor}. This is because, as shown in Figure \ref{5(2)}, the charge $Q$ remains almost constant, while the mass of the black hole monotonically decreases.
\begin{figure}[!h]
\centering
\includegraphics[width=0.80\textwidth]{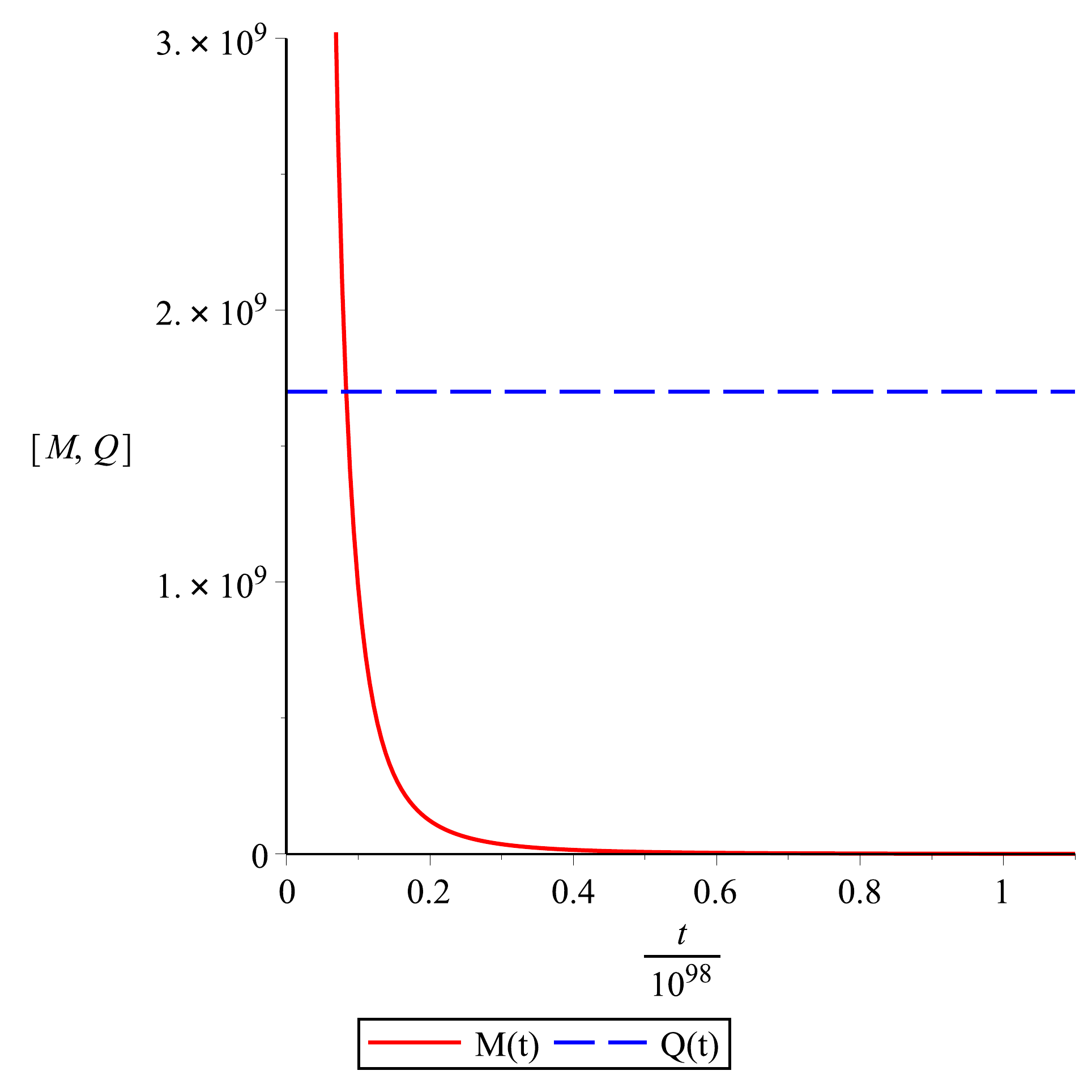}
\caption{The evolution of mass and charge of a toral black hole with $K=1$, $M(0)=5.6 \times 10^{20}$ cm, and initial charge $1.7\times10^{9}$ cm. Note that we are allowed to have $Q>M$ since the extremal black hole satisfies $Q=wM$ instead of $Q=M$. The charge $Q$ is not strictly constant, but drops by an amount too small to be noticeable at this scale. [Here, and everywhere henceforth, the units of $t$ are years.] \label{5(2)}} 
\end{figure}
An example is provided in Figure (\ref{3}), in which the initial $(\widetilde{Q}/M)^2$ ratio is tiny: $1.95 \times 10^{-21}$; yet the black hole evolves to be \emph{nearly} extremal. [Here, and henceforth, ``approaching extremality'' is conveniently defined as ``reaching $(\widetilde{Q}/M)^2=0.9$''.] This takes about $4 \times 10^{98}$ years, and it seems extremely likely that the time required actually to reach extremality is infinite\footnote{Note that even if the black hole did become extremal in finite time, we would still have the same problem -- its \emph{lifetime} still appears to be infinite.}. At this point, the [Bekenstein-Hawking] entropy is still extremely large [see Figure (\ref{3(2)}).], of the order $10^{90}$ in these units. The decoding time according to HH is exponential in numbers of this order, but it is still finite. This is our problem.

\begin{figure}[!h]
\centering
\mbox{\subfigure{\includegraphics[width=3in]{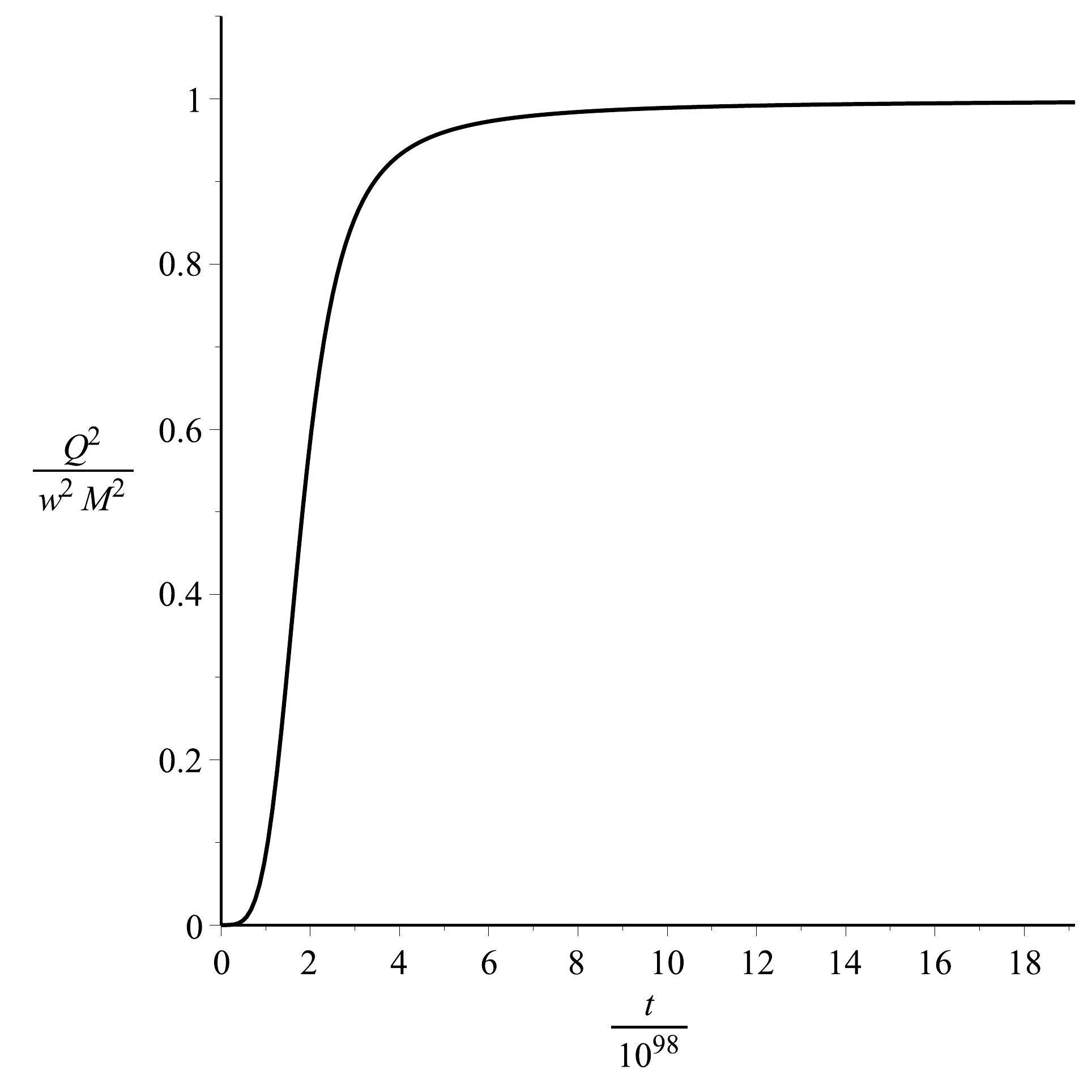}}\quad
\subfigure{\includegraphics[width=3in]{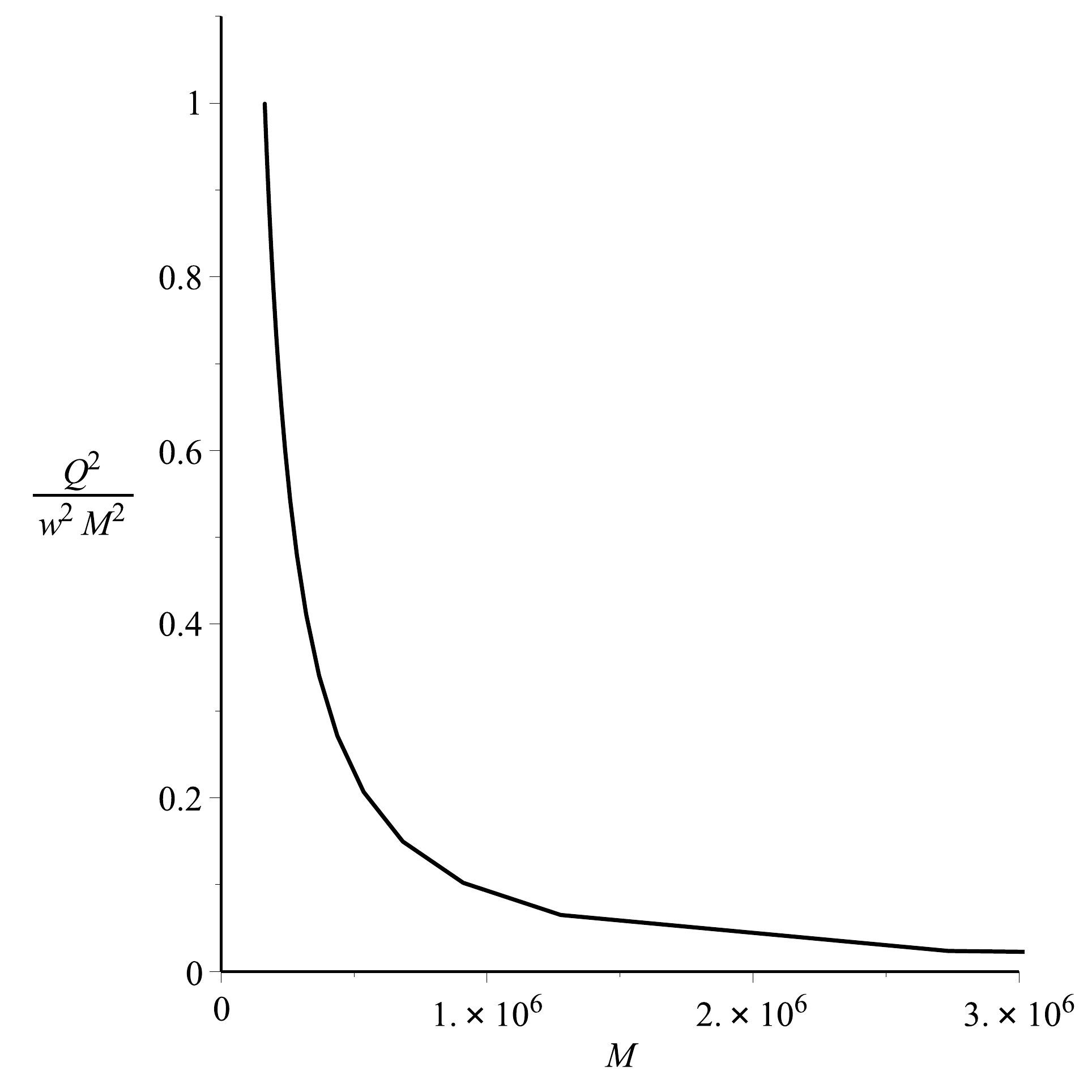} }}
\caption{\textbf{Left:} The square of the normalized charge-to-mass ratio as a function of time of a charged toral black hole with $K=1$, and initial condition $M(0)=5.6 \times 10^{20}$ cm, $Q(0)=1.7 \times 10^{9}$ cm. \textbf{Right:} The square of the normalized charge-to-mass ratio as a function of mass of the same black hole. \label{3}}
\end{figure}

\begin{figure}[!h]
\centering
\includegraphics[width=0.6\textwidth]{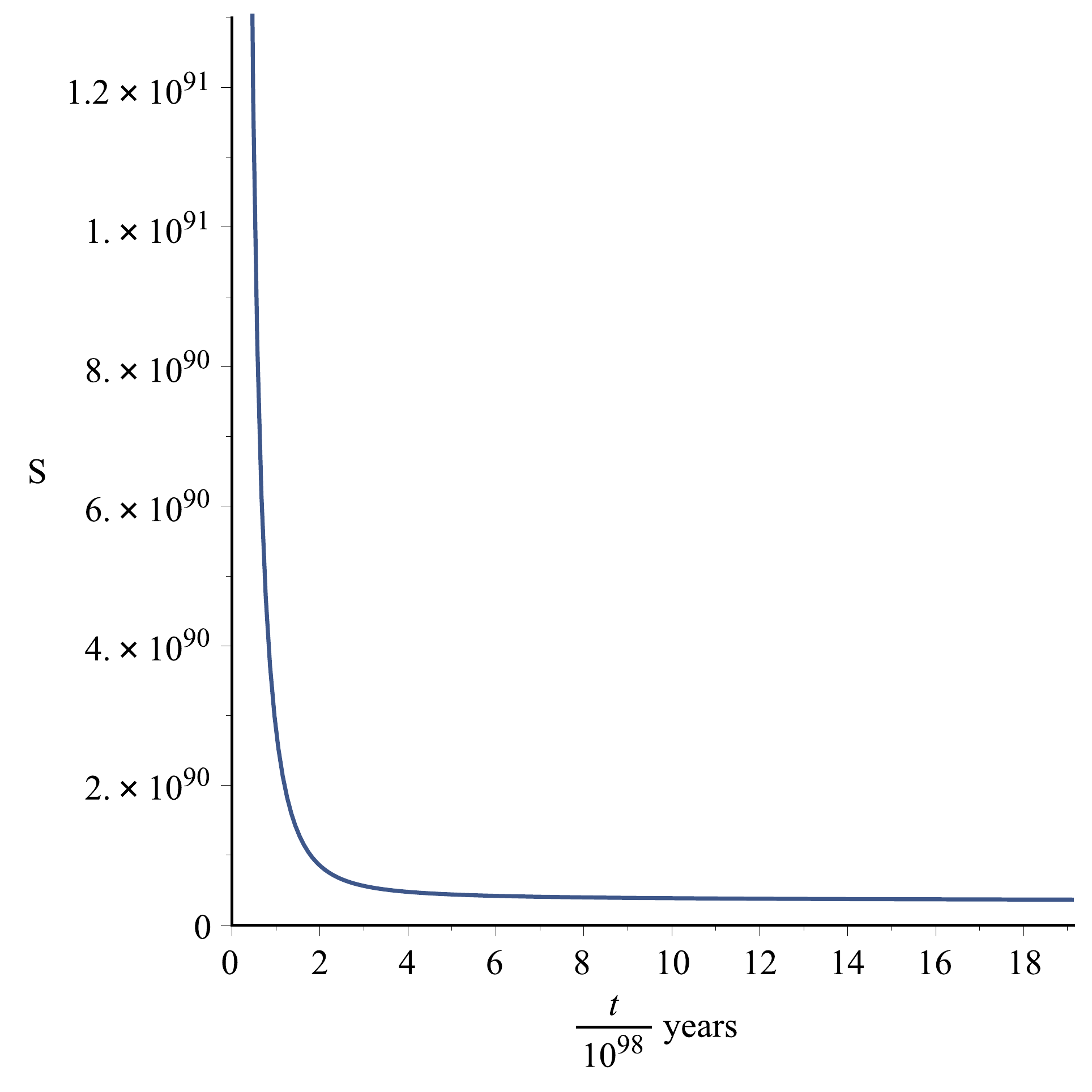}
\caption{The Bekenstein-Hawking entropy $S$, as a function of time, of a charged toral black hole with $K=1$, and initial conditions $M(0)=5.6 \times 10^{20}$ cm, $Q(0)=1.7 \times 10^{9}$ cm. Entropy, $S=A/(4\hbar)$ is a dimensionless number in our units, since $\hbar$ is an area.\label{3(2)}}
\end{figure}

The principal effect of varying the parameters is simply to modify the time scale of the attractor.
For example, a toral black hole with $K=1$, $M(0)=5.6 \times 10^{20}$ cm and $Q(0)=34.9 \times 10^{18}$ cm, that is, $(\widetilde{Q}/M)^2 = 0.82$ initially, takes about $10^{94}$ years to come close to $(\widetilde{Q}/M)^2 \approx 1$, while a black hole of the same mass, but with much lower charge, as shown in Figure (\ref{3}), takes about $10^{98}$ years. For the same initial mass and initial charge, increasing the values of $K$ lengthens the time required to approach extremality. This is shown in Figure \ref{4(2)} below. This is due to the fact that -- see equations \ref{w} and \ref{QwM} -- the initial [normalized] charge-to-mass ratio \emph{depends} on the choice of the compactification parameter $K$. On the other hand, increasing the value of $L$ extends the time it takes to approach extremality. For example, with the initial conditions $(M(0)=5.6 \times 10^{20} ~\text{cm}, ~Q(0)=1.7 \times 10^{9}~\text{cm})$, a charged toral black hole with $K=1,~L=10^{15}$ cm takes about $4 \times 10^{98}$ years to approach extremality, but, if we increase the value of $L$ to $10^{30}$ cm, the black hole now takes $10^{151}$ years to approach extremality; the time scale becomes $3 \times 10^{83}$ years if one decreases $L$ to $5 \times 10^{10}$ cm.

\begin{figure}[!h]
\centering
\includegraphics[width=0.6\textwidth]{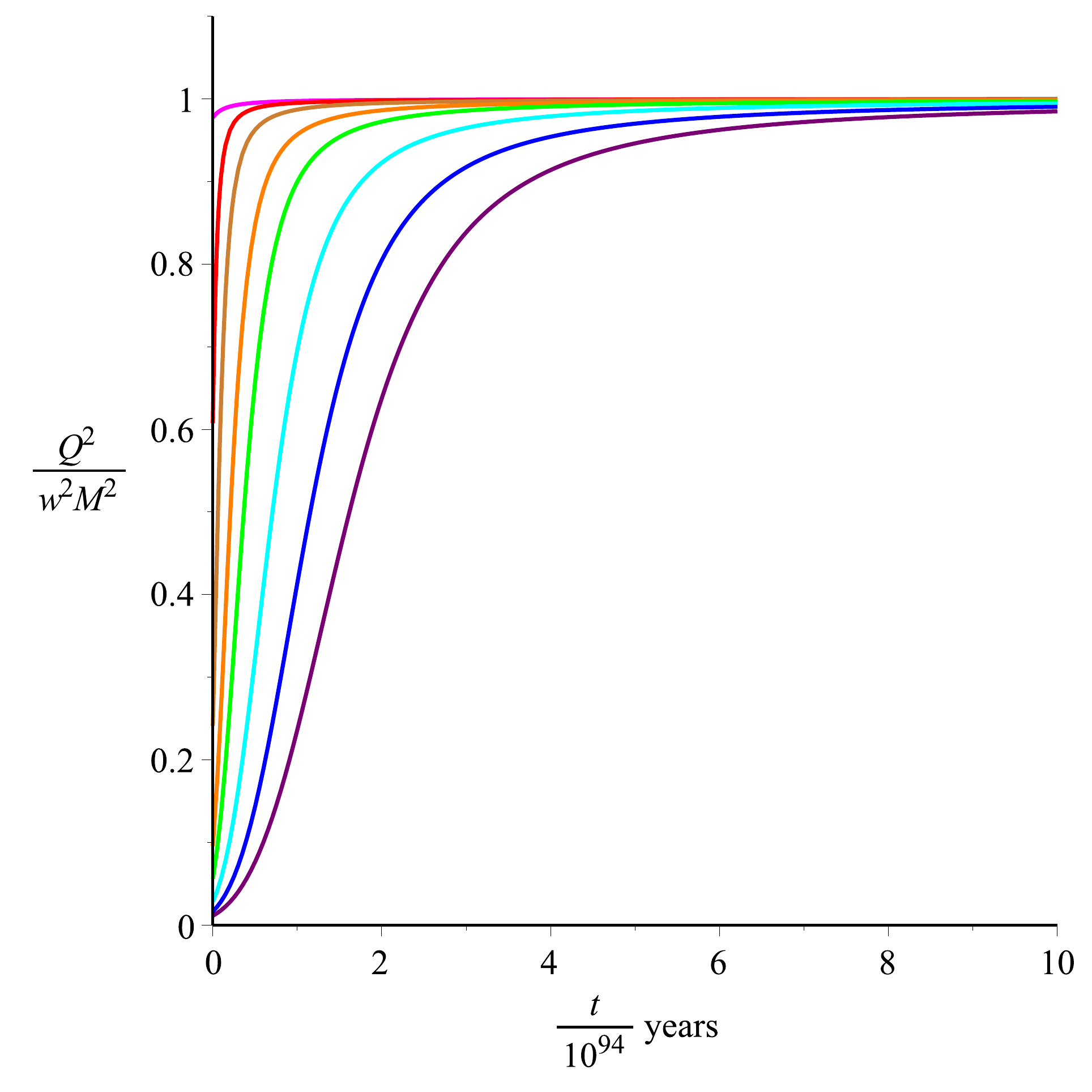}
\caption{The effect of varying the compactification parameter $K$ on the evolution of the normalized charge-to-mass ratio of a toral black hole, with initial mass and initial charge fixed to be $5.6\times 10^{20}$ cm and $3.0 \times 10^{19}$ cm, respectively. From top to bottom, the curves correspond to $K=0.7, 1, 2, 4, 6, 10, 15$ and 20, respectively. \label{4(2)}}
\end{figure}

Of course, starting with a lower value of the initial charge for a fixed initial mass also lengthens the time it takes to approach extremality.
An extreme example is shown in the left plot of Fig.(\ref{5}), in which we still keep the initial mass as $M(0)=5.6 \times 10^{20}$ cm, but set $Q(0)=6 \times 10^{-34}$ cm for a toral black hole with $K=1$. The black hole takes, as expected, a much longer time -- $~10^{119}$ years -- to approach extremality.
\begin{figure}[!h]
\centering
\mbox{\subfigure{\includegraphics[width=3in]{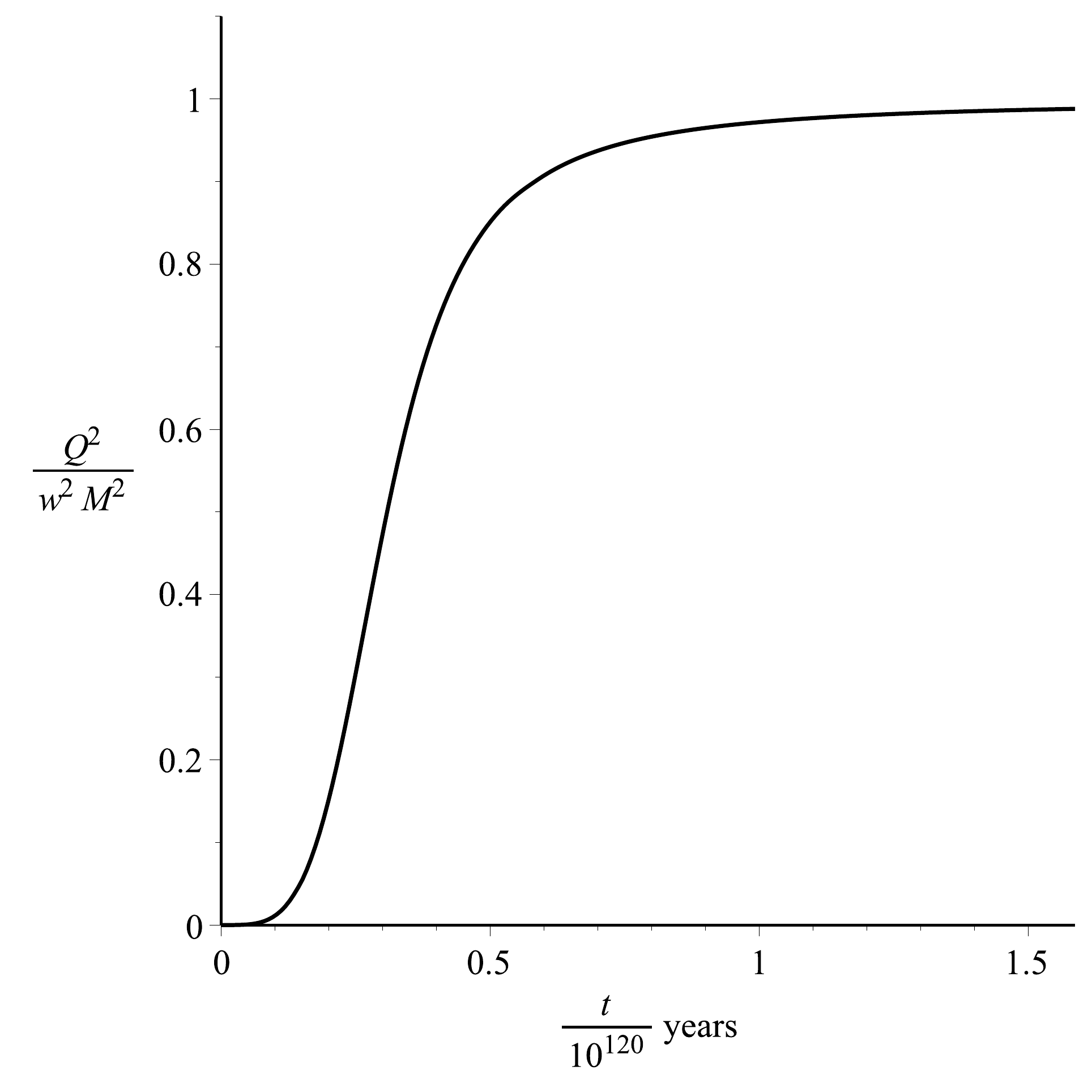}}\quad
\subfigure{\includegraphics[width=3in]{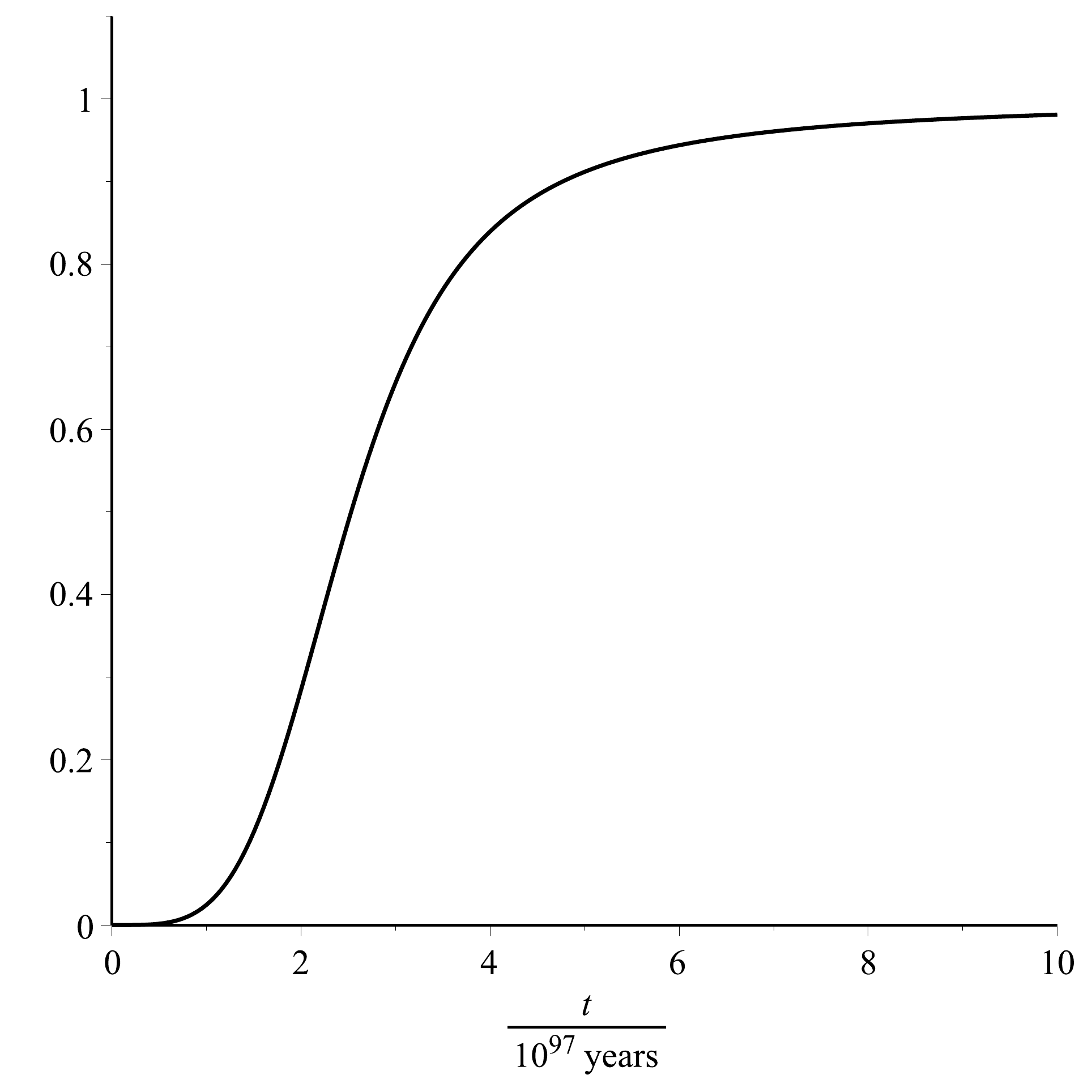} }}
\caption{\textbf{Left:} The square of the normalized charge-to-mass ratio of a toral black hole with $K=1$, initial mass $M(0)=5.6 \times 10^{20}$ cm, and initial charge $Q(0)=6 \times 10^{-34}$ cm. \textbf{Right:} The square of the normalized charge-to-mass ratio of a planar black hole with $M^*=5.6 \times 10^{30}$ cm and $Q^*=1.7 \times 10^9$ cm.  \label{5}}
\end{figure}

The results discussed above also hold for planar black holes -- one example is provided in the right plot of Fig.(\ref{5}). Thus we see that all toral and planar electrically charged AdS black holes are, as they evaporate, driven towards [and come \emph{arbitrarily close} to] extremality, on time scales which are short relative to the decoding time.

Of course, this statement does not apply to flat black holes which are \emph{exactly} electrically neutral. However, from our point of view here, such black holes should be considered unstable. If the black hole should acquire any amount of charge, \emph{no matter how small}, it will be swept away towards extremality by the evaporation process. Thus, physically, one should not regard this special case as an exception.

In short, then, a generic black hole with a flat event horizon will get steadily \emph{colder}. One might think that planar black holes, which are immune to the Hawking-Page transition discussed earlier, are therefore less at risk of being destroyed as time passes. That is not correct, as we now explain.

As the charge on any black hole increases, the geometry of the ambient spacetime changes. It follows that the geometry of any extended object in that ambient space is also affected. This is directly relevant to the AdS/CFT correspondence, because string theory in the AdS bulk does, of course, entail the existence of extended objects ---$\,$ branes. In particular, the action of a BPS brane depends on its area and its volume, and Seiberg and Witten \cite{kn:seiberg} showed that it is possible for modifications of the bulk geometry to distort the brane geometry in such a way that the consequent changes to the areas and volumes cause the brane action to become \emph{negative}. The resulting instability is a generalization of the black hole ``fragmentation'' effect on which HH hope to rely. The work of Seiberg and Witten allows us to be more explicit than was possible in \cite{kn:HH}.

Seiberg and Witten stressed that the situation is particularly delicate when the boundary geometry is [scalar-]flat ---$\,$ which is precisely the case here. In \cite{kn:AdSRN}\cite{universal} it was shown that electrically neutral AdS black holes with flat event horizons are stable in this sense, and in fact this remains true for most values of the electric charge below the extremal value. However, when the electric charge becomes sufficiently large \emph{but still sub-extremal}, the distortion of the branes does become large enough to trigger the instability. In four dimensions, this happens when the charge parameter is around 0.916 times the extremal value\footnote{In the case of an $(n+2)$-dimensional black hole, the instability is triggered when the electric charge exceeds $\sqrt{{n-1 \over n+1} \Bigg[{n\over n-1}\Bigg]^{{2n\over n+1}}}\times Q_{extremal}$. See \cite{universal} for detailed discussions.}.

Combining this with our findings in this work, we see that, as these black holes evaporate, they inevitably [unless they are destroyed in some other way first] come sufficiently close to extremality to trigger the Seiberg-Witten effect, and this happens in a time which is very short relative to the decoding time. That is, the black hole ceases to exist before its Hawking radiation can be decoded.

In the planar case, this is the only effect we need to consider, since there is no phase transition. In the toral case, however, it is possible for the hole to undergo a phase transition before the Seiberg-Witten instability arises, or vice versa. The question as to which effect actually destroys the hole can only be answered by considering each case in detail. One way to investigate this is to plot the normalized charge-to-mass ratio against the temperature, and see whether the black hole first reaches $(\widetilde{Q}/M)^2 \approx 0.84$ or $T_c$. In other words, the ultimate fate of a given black hole depends on the competition between the fall in temperature and the rise in the normalized charge-to-mass ratio. Two examples are provided in Figure (\ref{6}),
both of which describe black holes which are destroyed by the Seiberg-Witten effect.

\begin{figure}[!h]
\centering
\mbox{\subfigure{\includegraphics[width=3in]{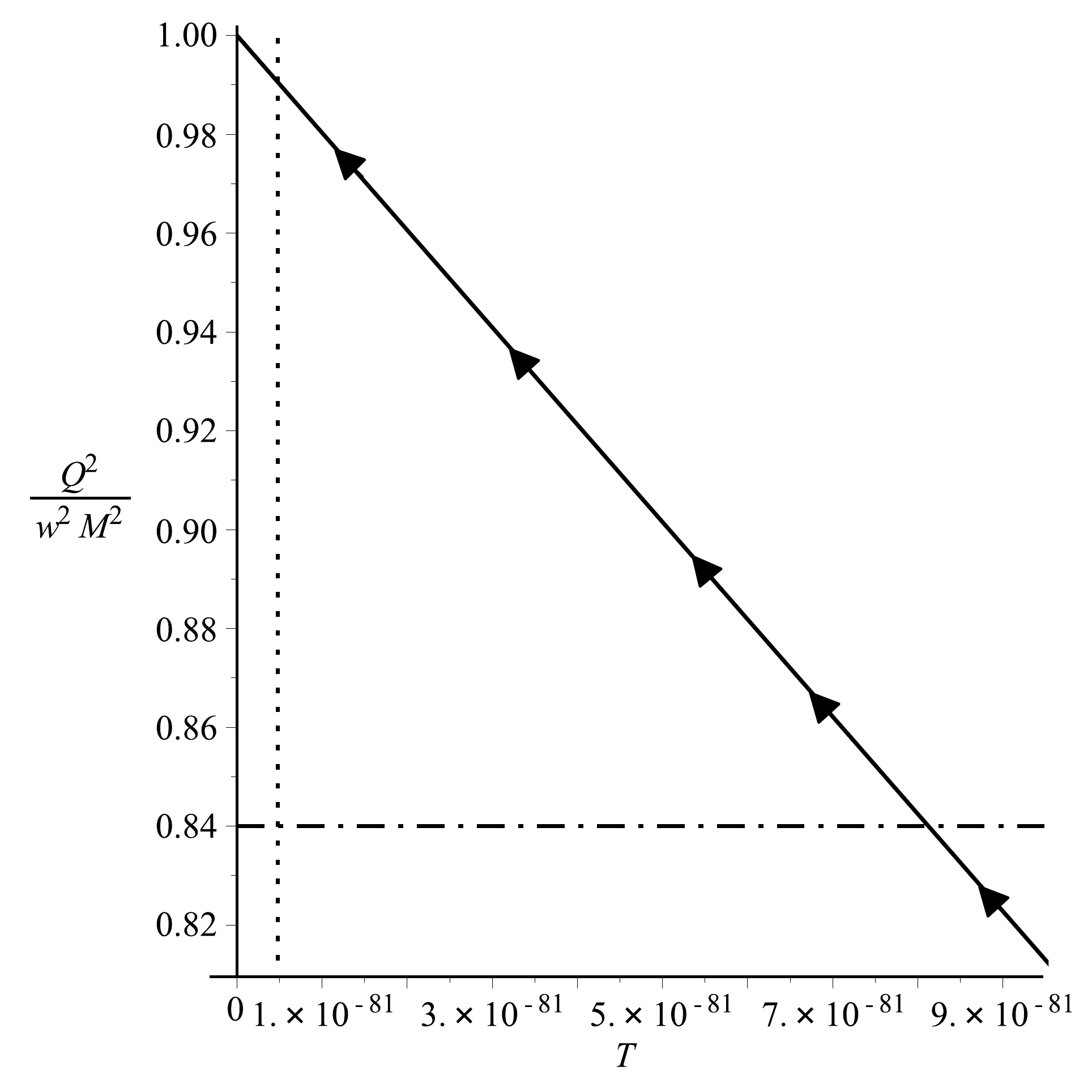}}\quad
\subfigure{\includegraphics[width=3in]{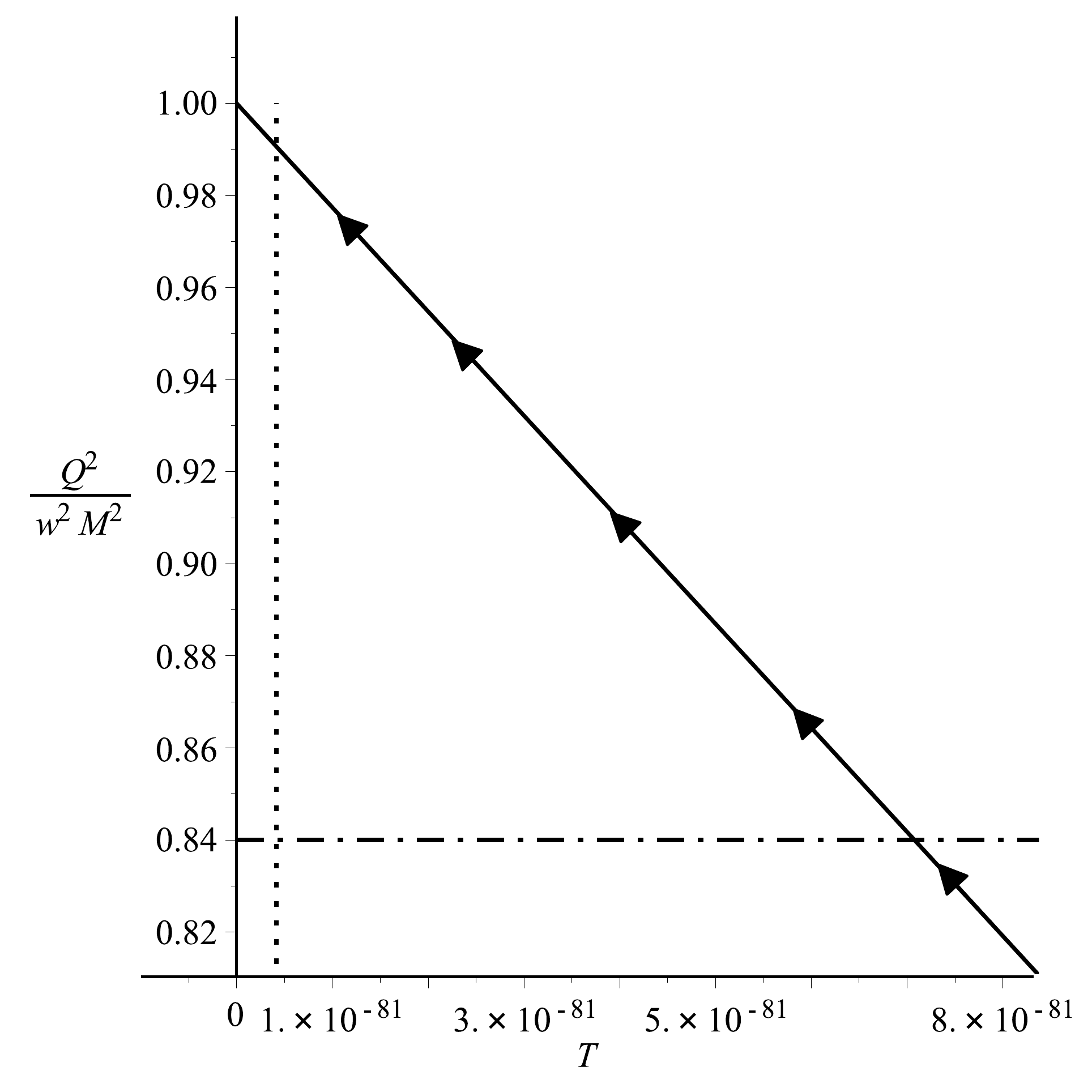} }}
\caption{The square of the normalized charge-to-mass ratio as a function of temperature for a charged toral black hole with $K=1$, and initial conditions $(M(0)=5.6 \times 10^{20} ~\text{cm}, ~Q(0)=3.0 \times 10^{9}~\text{cm})$ [\textbf{left}], and $(M(0)=5.6 \times 10^{20}~\text{cm}, ~Q(0)=6 \times 10^{-34}~\text{cm})$ [\textbf{right}], respectively.
Dotted lines indicate the critical temperature below which the black holes undergo phase transition into solitons; this, for $K=1$, is $T_c=4.16 \times 10^{-82}$ cm. The dot-dash lines indicate the threshold beyond which black holes become unstable due to the Seiberg-Witten effect. In this case, both black holes reach the dot-dash lines first. \label{6}}
\end{figure}

\begin{figure}[!h]
\centering
\includegraphics[width=0.5\textwidth]{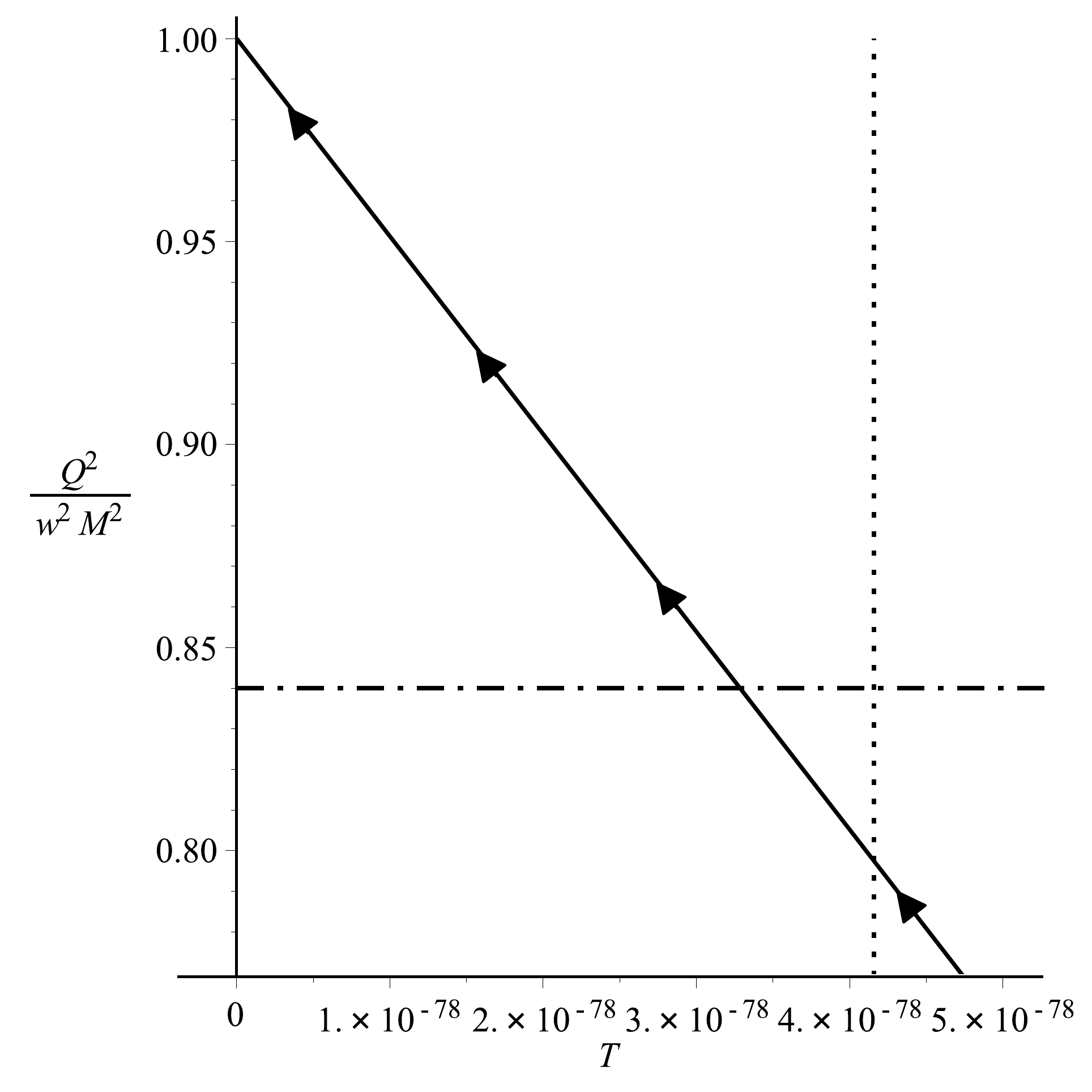}
\caption{The square of the normalized charge-to-mass ratio as a function of temperature of a charged toral black hole with $K=10^{-4}$, and initial condition $M(0)=5.6 \times 10^{20}~\text{cm}, ~Q(0)=6 \times 10^{-34}~\text{cm}$. The dot-dash line indicates the threshold beyond which black holes become unstable due to the Seiberg-Witten effect.
Dotted line indicates the critical temperature below which black holes undergo phase transition into soliton; which is $T_c=4.16 \times 10^{-78}$ cm for $K=10^{-4}$. In this example, the black hole reaches the dotted line first.
\label{7}}
\end{figure}

Since $T_c$ is controlled by $K$, we see that, for lower values of $K$, the black hole tends to be destroyed by a phase transition into a soliton, instead of by Seiberg-Witten instability: see for example Figure (\ref{7}), with $M(0)=5.6 \times 10^{20}~\text{cm}, Q(0)=6 \times 10^{-34}~\text{cm}$. These are the same initial conditions as in the right plot of Figure (\ref{6}), except that $K$ is now $10^{-4}$. The black hole now reaches the phase transition temperature first, before $Q^2/w^2M^2$ falls below 0.84.
On the other hand, for larger values of $K$, black holes tend to be destroyed by Seiberg-Witten instability rather than a phase transition.

\addtocounter{section}{1}
\section* {\large{\textsf{5. Conclusion: Hawking Radiation Cannot Be Decoded, Even if the Curvature is Small}}}
Attempts to settle the question of the unitarity of black hole evolution are plagued by uncertainties connected with quantum gravity, particularly in the high-curvature regime. This prompts the question: what can be said if we approach the problem while staying clear, as far as possible, of these uncertainties?

The AdS/CFT correspondence permits a \emph{definition} of a quantum-gravitational system in terms of a well-understood field theory at infinity. That field theory is maximally well-understood when the boundary geometry is just
flat spacetime. We therefore argue that the most reliable context for discussing these issues is provided by AdS black holes with flat event horizons, since these are dual to a field theory on a boundary which is either locally or even globally flat. We have shown that these black holes have the great virtue of evaporating towards extremality: that is, they become cold, and the curvature outside the event horizon remains very small at all times. Thus we simultaneously avoid the high-curvature regime, and probe the physics precisely in the regime where the firewall paradox is sharpest. We find that low temperatures tend to destroy such black holes, just as their duality with the quark-gluon plasma would suggest.

The destruction takes a long time by normal standards, like the evaporation of most black holes; but compared to the time required to ``decode the Hawking radiation'', it happens very quickly. In short, in the best-understood cases, \emph{Hawking radiation cannot be decoded}, confirming the claim of Harlow and Hayden.

It remains to be seen whether the fact that Hawking radiation cannot be decoded really resolves the firewall problem; we hope that our results will stimulate renewed efforts to overcome the objections raised in \cite{kn:apologia}.

\addtocounter{section}{1}
\section*{\large{\textsf{Acknowledgement}}}
Yen Chin Ong thanks Wen-Yu Wen, Keisuke Izumi, Sean Downes, Je-An Gu, Ue-Li Pen, Haret C. Rosu, Chiang-Mei Chen and Don Page for very useful discussions on various aspects of black hole physics. He also thanks Ian Anderson for conducting the MAPLE workshop during the GR20/Amaldi10 conference in Warsaw, Poland, which proved useful for this work. Brett McInnes thanks Soon Wanmei and Jude McInnes for helpful advice.
This work is supported by Taiwan National Science Council, Taiwan's National Center for Theoretical Sciences [NCTS], and the Leung Center for Cosmology and Particle Astrophysics [LeCosPA] of National Taiwan University.

\end{document}